\documentclass[11pt, a4paper, reqno]{article}
\usepackage{amsmath,amsthm,amsfonts,amssymb}
\usepackage[english]{babel}
\usepackage{a4wide}
\usepackage{microtype}
\usepackage{bm}
\usepackage{graphics}
\usepackage{epic}
\usepackage{hyperref}
\hypersetup{colorlinks=true,citecolor=blue,linkcolor=blue,urlcolor=blue}
\usepackage{url}

\newcommand{\PatternTOne}{T1}
\newcommand{\PatternTTwo}{T2}
\newcommand{\PatternTThree}{T3}
\newcommand{\PatternTFour}{T4}
\newcommand{\PatternTFive}{T5}
\newcommand{\PatternPThree}{Q1}
\newcommand{\PatternQThree}{Q2}

\newcommand{\PatternRFive}{R5}

\newcommand{\PatternRSeven}{R7}
\newcommand{\PatternRSevenMinus}{R7-}
\newcommand{\PatternREight}{R8}

\newcommand{\PatternMHat}{\^M}
\newcommand{\PatternVTwo}{V$_2$}
\newcommand{\CSP}[1]{{\sc CSP(}$\overline{\mbox{#1}}${\sf)}}
\newcommand{\ignore}[1]{}

\def\mvar{X}
\def\mdom{D}
\def\mcons{C}

\def\ixv{I_{xv}}

\def\pxv{(P_{xv})}
\def\spxv{S_{\pxv}}
\def\sipxv{S^I_{\pxv}}
\newcommand{\dxv}[1]{D_{xv}(#1)}

\newtheorem{theorem}{Theorem}
\newtheorem*{theorem*}{Theorem}
\newtheorem{lemma}{Lemma}
\newtheorem{proposition}{Proposition}
\theoremstyle{definition}
\newtheorem{remark}{Remark}

\begin{document}

\title{On Singleton Arc Consistency for CSPs\\ Defined by Monotone
Patterns\footnote{An extended abstract of this work appeared in Proceedings of the
35th International Symposium on Theoretical Aspects of Computer Science
(STACS)~\cite{cccz18:stacs}. The authors were supported by EPSRC grant EP/L021226/1 and ANR-11-LABX-0040-CIMI within the program ANR-11-IDEX-0002-02. Stanislav \v{Z}ivn\'y was supported by a Royal Society University Research Fellowship. This project has received funding from the European Research Council (ERC) under the European Union's Horizon 2020 research and innovation programme (grant agreement No 714532). The paper reflects only the authors' views and not the views of the ERC or the European Commission. The European Union is not liable for any use that may be made of the information contained therein.}}

\author{
Cl\'ement Carbonnel\\
University of Oxford\\
\texttt{clement.carbonnel@cs.ox.ac.uk}
\and
David A. Cohen\\
Royal Holloway, University of London\\
\texttt{d.cohen@rhul.ac.uk}
\and
Martin C. Cooper\\
IRIT, University of Toulouse III\\
\texttt{cooper@irit.fr}
\and
Stanislav \v{Z}ivn\'{y}\\
University of Oxford\\
\texttt{standa.zivny@cs.ox.ac.uk}
}

\date{}
\maketitle

\begin{abstract}

Singleton arc consistency is an important type of local consistency
which has been recently shown to solve all constraint satisfaction
problems (CSPs) over constraint languages of bounded width. We aim to
characterise all classes of CSPs defined by a forbidden pattern that
are solved by singleton arc consistency and closed under removing
constraints. We identify five new patterns whose absence ensures
solvability by singleton arc consistency, four of which are provably maximal
and three of which generalise 2-SAT. Combined with simple
counter-examples for other patterns, we make significant
progress towards a complete classification.

\end{abstract}

\section{Introduction}
\label{sec:Introduction} The constraint satisfaction problem (CSP) is
a declarative paradigm for expressing computational problems. An
instance of the CSP consists of a number of variables to which we
need to assign values from some domain.  Some subsets of the
variables are constrained in that they are not permitted to take all
values in the product of their domains.  The scope of a constraint is the set of variables whose
values are limited by the constraint, and the constraint relation is the set of permitted assignments to the variables of the scope.  A solution to a
CSP instance is an assignment of values to variables in such a
way that every constraint is satisfied, i.e. every scope is assigned to an element of the constraint relation.

The CSP has proved to be a useful technique for modelling in many
important application areas from manufacturing to process
optimisation, for example planning and
scheduling optimisation~\cite{Ozturk201626}, resource
allocation~\cite{lim2014constraint}, job shop
problems~\cite{Cheng1997} and workflow
management~\cite{Senkul2005399}. Hence much work has been done on
describing useful classes of constraints~\cite{globcon:2017} and
implementing efficient algorithms for processing
constraints~\cite{Brailsford1999557}.  Many constraint solvers use a
form of backtracking where successive variables are assigned values
that satisfy all constraints.  In order to mitigate the exponential
complexity of backtracking some form of pre-processing is always
performed.  These pre-processing techniques identify values that cannot be part of any solution in an effective way
and then propagate the effects of removing these values throughout
the problem instance.  Of key importance amongst these pre-processing
algorithms are the relatives of arc consistency propagation including
generalised arc consistency (GAC) and singleton arc consistency
(SAC). Surprisingly
there are large classes~\cite{Cohen17:GAC,cz17:lmcs,chen13:sac,kozik16:lics} of the CSP for which GAC or SAC are decision
procedures: after establishing consistency if every variable still
has a non-empty domain then the instance has a solution.  

More generally, these results fit into the wider area of research aiming to identify sub-problems of the CSP for which certain polynomial-time algorithms are decision procedures. Perhaps the most natural ways to restrict the CSP is to limit the constraint relations that we allow or to limit the structure of (the hypergraph of) interactions of the constraint scopes.  A set of allowed constraint relations is called a constraint language.  A subset of the CSP defined by limiting the scope interactions is called a structural class. 

There has been considerable success in identifying tractable
constraint languages, recently yielding a full classification of the complexity of finite constraint languages~\cite{focs/Bulatov17a,focs/Zhuk17}.  Techniques from universal algebra have been essential in this work as the complexity of a constraint language is characterised by a particular algebraic
structure~\cite{Bulatov:2005}. The two most important algorithms for solving the CSP over tractable constraint languages are local consistency and the few subpowers algorithm~\cite{Bulatov:2006,idziak10:sicomp}, which generalises ideas from group theory. A necessary and sufficient condition for solvability by the few subpowers algorithm was identified in~\cite{idziak10:sicomp,berman2010varieties}. The set of all constraint languages decided by local consistency was later described by Barto and Kozik~\cite{Barto14:jacm} and independently by Bulatov~\cite{bulatov2009bounded}. Surprisingly, all such languages are in fact decided by establishing singleton arc consistency~\cite{kozik16:lics}.

A necessary condition for the tractability of a structural class with
bounded arity is that it has bounded treewidth modulo homomorphic equivalence~\cite{Grohe07:jacm}. In all such cases we decide an instance by establishing $k$-consistency, where $k$ is the treewidth of the core. It was later shown that the converse holds: if a class of structures does not have treewidth $k$ modulo homomorphic equivalence then it is not solved by $k$-consistency~\cite{atserias2007power}, thus fully characterising the strength of consistency algorithms for structural restrictions. Both language-restricted CSPs and CSPs of bounded treewidth are \emph{monotone} in the sense that we can relax (remove constraints from) any CSP instance without affecting its membership in such a class.

Since our understanding of consistency algorithms for language and structural classes is so well advanced there is now much interest in so called hybrid classes, which are neither definable by restricting the language nor by limiting the structure. For the binary CSP, one popular mechanism for defining hybrid classes follows the considerable success of mapping the complexity landscape for graph problems in the absence of certain induced subgraphs or graph minors. Here, hybrid (binary) CSP problems are defined by forbidding a fixed set of substructures (\emph{patterns}) from occurring in the instance~\cite{cccms12:jair}. This framework is particularly useful in algorithm analysis, since it allows us to identify precisely the local properties of a CSP instance that make it impossible to solve via a given polynomial-time algorithm. This approach has recently been used to obtain a pattern-based characterisation of solvability by arc consistency~\cite{cz17:lmcs}, a detailed analysis of variable elimination rules~\cite{ccez15:jcss} and various novel tractable classes of CSP~\cite{Cooper15:dam,cooper16:ai}.

Singleton arc consistency is a prime candidate to study in this framework since it is one of the most prominent incomplete polynomial-time algorithms for CSP and the highest level of consistency (among commonly studied consistencies) that operates only by removing values from domains. This property ensures that enforcing SAC cannot introduce new patterns, which greatly facilitates the analysis. It is therefore natural to ask for which patterns, forbidding their occurrence ensures that SAC is a sound decision procedure. In this paper we make a significant contribution towards this objective by identifying five patterns which define classes of CSPs decidable by SAC. All five classes are monotone, and we show that only a handful of open cases separates us from an essentially full characterisation of monotone CSP classes decidable by SAC and definable by a forbidden pattern. Some of our results rely on a novel proof technique which follows the \emph{trace} of a successful run of the SAC algorithm to dynamically identify redundant substructures in the instance and construct a solution.

The structure of the paper is as follows. In
Section~\ref{sec:Preliminaries} we provide essential definitions and
background theory.  In Section~\ref{sec:Results} we state the main
results. In Section~\ref{sec:trace} we introduce the trace technique, which is then used
in Sections~\ref{sec:PatternPThree} and~\ref{sec:R87-} to establish the tractability
of three patterns from our main result. The tractability of the remaining two
patterns from the main result is shown in Section~\ref{sec:Q3R5}. In
Section~\ref{sec:CounterExamples} we give a necessary condition for the
solvability by SAC. Finally, we conclude the paper in
Section~\ref{sec:Conclusion} with some open problems.

\section{Preliminaries}
\label{sec:Preliminaries}

\subparagraph*{CSP}
A \emph{binary CSP instance} is a triple $I=(X,D,C)$, where $X$ is a
finite set of variables, $D$ is a finite domain, each variable $x\in
X$ has its own domain of possible values $D(x) \subseteq D$, and
$C=\{R(x,y)\mid x, y\in X, x\neq y\}$, where $R(x,y)\subseteq D^2$, is the set
of constraints. We assume, without loss of generality,
that each pair of variables $x,y \in X$ is constrained by a
constraint $R(x,y)$. (Otherwise we set $R(x,y)=D(x) \times D(y)$.) We
also assume that $(a,b) \in R(x,y)$ if and only if $(b,a) \in
R(y,x)$. A constraint is \emph{trivial} if it contains the Cartesian product of the domains of the two variables.  By \emph{deleting} a constraint we mean replacing it with a trivial constraint.   The
\emph{projection} $I[X']$ of a binary CSP instance $I$ on
$X'\subseteq X$ is obtained by removing all variables in $X
\backslash X'$ and all constraints $R(x,y)$ with $\{x,y\}
\not\subseteq X'$. A \emph{partial solution} to a binary CSP instance
on $X'\subseteq X$ is an assignment $s$ of values to variables in
$X'$ such that $s(x) \in D(x)$ for all $x \in X'$ and $(s(x),s(y))
\in R(x,y)$ for all constraints $R(x,y)$ with $x,y \in X'$. A
\emph{solution} to a binary CSP instance is a partial solution on
$X$.

An assignment $(x,a)$ is called a \emph{point}. For simplicity of notation we can assume that variable domains
are disjoint, so that using $a$ as a shorthand for $(x,a)$ is
unambiguous. If $(a,b) \in
R(x,y)$, we say that the assignments $(x,a),(y,b)$ (or more simply
$a,b$) are \emph{compatible} and that $ab$ is a \emph{positive edge},
otherwise $a,b$ are \emph{incompatible} and $ab$ is a \emph{negative
edge}. We say that $a \in D(x)$ has a \emph{support} at
variable $y$ if $\exists b \in D(y)$ such that $ab$ is a positive
edge.

The constraint graph of a CSP instance with variables $X$ is the graph
$G=(X,E)$ such that $(x,y)\in E$ if $R(x,y)$ is non-trivial. The \emph{degree}
of a variable $x$ in a CSP instance is the degree of $x$ in the constraint graph
of the instance.
\subparagraph*{Arc Consistency}
A domain value $a \in D(x)$ is \emph{arc consistent} if it has a support
at every other variable. A CSP instance is \emph{arc consistent} (AC) if every domain value is arc consistent. 
\subparagraph*{Singleton Arc Consistency}
Singleton arc consistency is stronger than arc consistency (but
weaker than strong path consistency~\cite{MOHR1986225}). A domain value $a \in D(x)$ in a CSP instance
$I$ is \emph{singleton arc consistent} if the instance obtained from $I$ by
removing all domain values $b \in D(x)$ with $a\neq b$ can be made arc consistent without emptying any domain. A
CSP instance is \emph{singleton arc consistent} (SAC) if every domain value is
singleton arc consistent.  
\subparagraph*{Establishing Consistency}
Domain values that are not arc consistent or not singleton arc consistent cannot be part of a solution so can
safely be removed. The closure of a CSP instance under the removal of values that are not (singleton) arc consistent is unique, and the process of reducing an instance to its closure is called \emph{establishing (singleton) arc consistency}~\cite{Rossi:handbook}. For a binary CSP instance with domain size $d$,  $n$ variables and $e$ non-trivial constraints there are $O(ed^2)$ algorithms for establishing arc consistency~\cite{Bessiere05:ai} and $O(ned^3)$ algorithms for establishing singleton arc
consistency~\cite{Bessiere08:ai}.

SAC \emph{decides} a CSP instance if, after establishing
singleton arc consistency, non-empty domains for all variables guarantee the
existence of a solution. SAC decides a class of CSP
instances if SAC decides every instance from the class.
\subparagraph*{Neighbourhood Substitutability}
If $a,b \in D(x)$, then $a$ is \emph{neighbourhood substitutable} (or is \emph{dominated}) by
$b$ if there is no $c$ such that $ac$ is a positive edge and $bc$ a
negative edge: such values $a$ can be deleted from $D(x)$ without
changing the satisfiability of the instance since $a$ can be replaced
by $b$ in any solution~\cite{Freuder91:interchangeable}. Similarly, removing neighbourhood substitutable
values cannot destroy (singleton) arc consistency.

\subparagraph{Patterns}
In a binary CSP instance each constraint decides, for each pair of values in $D$, whether it is allowed.  Hence a binary CSP can also be defined as a set of points $\mvar \times D$ together with a compatibility function that maps each edge, $(x,a)(y,b)$ with $x \neq y$, into the set $\{\mbox{negative},\mbox{positive}\}$.  A \emph{pattern} extends the notion of a binary CSP instance by allowing the compatibility function to be partial.
A pattern $P$ \emph{occurs} (as a subpattern) in an instance $I$ if there is mapping from the points of $P$ to the points of $I$
which respects variables (two points are mapped to points of the same variable in $I$ if and only if they belong to the same variable in $P$) and maps negative edges to negative edges, and positive edges to positive edges. A set of patterns occurs in an instance $I$ if at least one pattern in the set occurs in $I$.

We use the notation \CSP{$P$} for the set of binary instances in which $P$ does
not occur as a subpattern. A pattern $P$ is \emph{SAC-solvable} if SAC decides \CSP{$P$}. It is worth observing that \CSP{$P$} is closed under the operation of establishing (singleton) arc
consistency. A pattern $P$ is \emph{tractable} if \CSP{$P$} can be solved in polynomial time.

Points $(x,a)$ and $(x,b)$ in a pattern are \emph{mergeable} if there
is no point $(y,c)$ such that $ac$ is positive and $bc$ is negative
or vice versa. For each set of patterns there exists a set of patterns without mergeable points which occur in the same set of instances.

A point $(x,a)$ in a pattern is called \emph{dangling}
if there is at most one $b$ such that $ab$ is a positive edge and no
$c$ such that $ac$ is a negative edge. Dangling points are redundant
when considering the occurrence of a pattern in an arc consistent CSP instance. 

A pattern is called
\emph{irreducible} if it has no dangling points and no mergeable
points~\cite{Cooper15:dam}. When studying algorithms that are at least as strong as arc consistency, a classification with respect to forbidden \emph{sets} of irreducible patterns is equivalent to a classification with respect to all forbidden sets of patterns. For this reason classifications are often established with respect to irreducible patterns even if only classes definable by forbidding a single pattern are considered~\cite{Cooper15:dam,cz17:lmcs}, as we do in the present paper.

\section{Results} \label{sec:Results}

Call a class $\mathcal{C}$ of CSP instances \emph{monotone} if deleting any
constraint from an instance $I \in \mathcal{C}$ produces another instance in
$\mathcal{C}$. For example, language classes and bounded treewidth classes are
monotone.
An interesting research direction is to study those monotone classes defined by a
forbidden pattern which are solved by singleton arc consistency, both in order
to uncover new tractable classes and to better understand the strength of SAC.

We call a pattern \emph{monotone} if when forbidden it defines a monotone class.
Monotone patterns can easily be seen to correspond to exactly those patterns in
which positive edges only occur in constraints which have at least one negative
edge.
To see this, firstly let $P$ be a pattern in which positive edges only occur in
constraints which have at least one negative edge. Note that deleting a constraint in an instance $I$ cannot introduce $P$, so \CSP{$P$} is monotone. To see the converse, let $Q$ be a pattern in which a positive edge $e$ occurs in a constraint $c$ with no negative edges. Let $Q'$ be equivalent to $Q$ but with $e$ replaced by a negative edge. Let $I'$ be the instance obtained by completing $Q'$ with negative edges (i.e. joining by negative edges all pairs of points at different variables whose compatibility is unspecified in $Q'$). Let $I'[-c]$ be the instance $I'$ in which the constraint (corresponding to) $c$ has been deleted. Now $Q$ occurs in $I'[-c]$ (since the positive edge $e$ has been reintroduced by deleting $c$) but not in $I'$ (which can be seen by simply counting the number of constraints containing positive edges). Thus \CSP{$Q$} is not monotone.

\thicklines \setlength{\unitlength}{1.5pt}

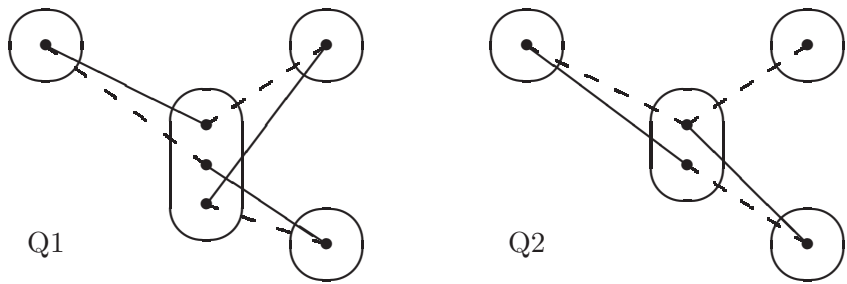
\begin{figure}
\centering
\begin{picture}(220,70)(0,80)

\put(0,80){\begin{picture}(90,70)(0,0)
\put(10,60){\makebox(0,0){$\bullet$}} \dashline{4}(10,60)(50,30) \dashline{4}(50,40)(80,60)
\put(80,60){\makebox(0,0){$\bullet$}} \put(10,60){\line(2,-1){40}} \dashline{4}(50,20)(80,10)
\put(50,20){\makebox(0,0){$\bullet$}} \put(50,20){\line(3,4){30}}  \put(50,30){\line(3,-2){30}}
\put(50,30){\makebox(0,0){$\bullet$}} \put(50,40){\makebox(0,0){$\bullet$}} \put(80,10){\makebox(0,0){$\bullet$}}
\put(10,60){\oval(18,18)} \put(50,30){\oval(18,38)} \put(80,10){\oval(18,18)} \put(80,60){\oval(18,18)}

\put(10,10){\makebox(0,0){\PatternPThree}}
\end{picture}}

\put(120,80){\begin{picture}(90,70)(0,0)
\put(10,60){\makebox(0,0){$\bullet$}} \dashline{4}(10,60)(50,40) \dashline{4}(50,40)(80,60)
\put(80,60){\makebox(0,0){$\bullet$}} \put(10,60){\line(4,-3){40}} \dashline{4}(50,30)(80,10)
\put(50,40){\line(1,-1){30}}
\put(50,30){\makebox(0,0){$\bullet$}} \put(50,40){\makebox(0,0){$\bullet$}} \put(80,10){\makebox(0,0){$\bullet$}}
\put(10,60){\oval(18,18)} \put(50,35){\oval(18,28)} \put(80,10){\oval(18,18)} \put(80,60){\oval(18,18)}

\put(10,10){\makebox(0,0){\PatternQThree}}
\end{picture}}

\end{picture}
\caption{All degree-3 irreducible monotone patterns solved by SAC must
occur in at least one of these patterns.} \label{fig:open3}
\end{figure}

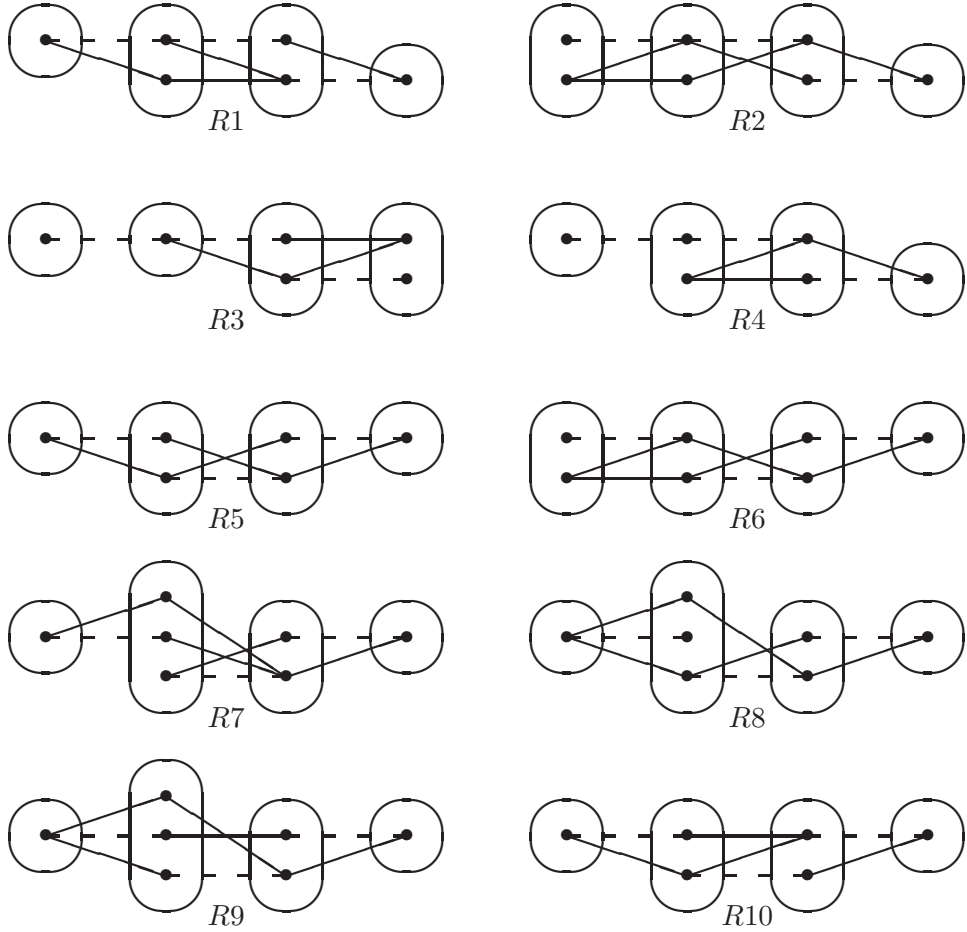
\begin{figure}
\centering
\begin{picture}(240,250)(0,0)

\put(0,200){\begin{picture}(110,40)(0,0)
\put(10,20){\makebox(0,0){$\bullet$}}
\put(40,10){\makebox(0,0){$\bullet$}} \put(40,20){\makebox(0,0){$\bullet$}}
\put(70,10){\makebox(0,0){$\bullet$}} \put(70,20){\makebox(0,0){$\bullet$}}
\put(100,10){\makebox(0,0){$\bullet$}}
\put(10,20){\oval(18,18)} \put(40,15){\oval(18,28)} \put(70,15){\oval(18,28)} \put(100,10){\oval(18,18)}
\dashline{4}(10,20)(40,20) \dashline{4}(40,20)(70,20) \dashline{4}(70,10)(100,10)
\put(10,20){\line(3,-1){30}} \put(40,20){\line(3,-1){30}} \put(70,20){\line(3,-1){30}}
\put(40,10){\line(1,0){30}}

\put(55,0){\makebox(0,0){$R1$}}
\end{picture}}

\put(130,200){\begin{picture}(110,40)(0,0)
\put(10,10){\makebox(0,0){$\bullet$}} \put(10,20){\makebox(0,0){$\bullet$}}
\put(40,10){\makebox(0,0){$\bullet$}} \put(40,20){\makebox(0,0){$\bullet$}}
\put(70,10){\makebox(0,0){$\bullet$}} \put(70,20){\makebox(0,0){$\bullet$}}
\put(100,10){\makebox(0,0){$\bullet$}}
\put(10,15){\oval(18,28)} \put(40,15){\oval(18,28)} \put(70,15){\oval(18,28)} \put(100,10){\oval(18,18)}
\dashline{4}(10,20)(40,20) \dashline{4}(40,20)(70,20) \dashline{4}(70,10)(100,10)
\put(10,10){\line(3,1){30}} \put(40,20){\line(3,-1){30}} \put(70,20){\line(3,-1){30}}
\put(40,10){\line(3,1){30}} \put(10,10){\line(1,0){30}}

\put(55,0){\makebox(0,0){$R2$}}
\end{picture}}

\put(0,150){\begin{picture}(110,40)(0,0)
\put(10,20){\makebox(0,0){$\bullet$}} \put(40,20){\makebox(0,0){$\bullet$}}
\put(70,10){\makebox(0,0){$\bullet$}} \put(70,20){\makebox(0,0){$\bullet$}}
\put(100,10){\makebox(0,0){$\bullet$}} \put(100,20){\makebox(0,0){$\bullet$}}
\put(10,20){\oval(18,18)} \put(40,20){\oval(18,18)} \put(70,15){\oval(18,28)} \put(100,15){\oval(18,28)}
\dashline{4}(10,20)(40,20) \dashline{4}(40,20)(70,20) \dashline{4}(70,10)(100,10)
 \put(40,20){\line(3,-1){30}} \put(70,10){\line(3,1){30}}
 \put(70,20){\line(1,0){30}}

\put(55,0){\makebox(0,0){$R3$}}
\end{picture}}

\put(130,150){\begin{picture}(110,40)(0,0)
\put(10,20){\makebox(0,0){$\bullet$}}
\put(40,10){\makebox(0,0){$\bullet$}} \put(40,20){\makebox(0,0){$\bullet$}}
\put(70,10){\makebox(0,0){$\bullet$}} \put(70,20){\makebox(0,0){$\bullet$}}
\put(100,10){\makebox(0,0){$\bullet$}}
\put(10,20){\oval(18,18)} \put(40,15){\oval(18,28)} \put(70,15){\oval(18,28)} \put(100,10){\oval(18,18)}
\dashline{4}(10,20)(40,20) \dashline{4}(40,20)(70,20) \dashline{4}(70,10)(100,10)
 \put(40,10){\line(1,0){30}} \put(70,20){\line(3,-1){30}}
\put(40,10){\line(3,1){30}}

\put(55,0){\makebox(0,0){$R4$}}
\end{picture}}

\put(0,100){\begin{picture}(110,40)(0,0)
\put(10,20){\makebox(0,0){$\bullet$}}
\put(40,10){\makebox(0,0){$\bullet$}} \put(40,20){\makebox(0,0){$\bullet$}}
\put(70,10){\makebox(0,0){$\bullet$}} \put(70,20){\makebox(0,0){$\bullet$}}
\put(100,20){\makebox(0,0){$\bullet$}}
\put(10,20){\oval(18,18)} \put(40,15){\oval(18,28)} \put(70,15){\oval(18,28)} \put(100,20){\oval(18,18)}
\dashline{4}(10,20)(40,20) \dashline{4}(40,10)(70,10) \dashline{4}(70,20)(100,20)
\put(10,20){\line(3,-1){30}} \put(40,20){\line(3,-1){30}} \put(70,10){\line(3,1){30}}
\put(40,10){\line(3,1){30}}

\put(55,0){\makebox(0,0){$R5$}}
\end{picture}}

\put(130,100){\begin{picture}(110,40)(0,0)
\put(10,10){\makebox(0,0){$\bullet$}} \put(10,20){\makebox(0,0){$\bullet$}}
\put(40,10){\makebox(0,0){$\bullet$}} \put(40,20){\makebox(0,0){$\bullet$}}
\put(70,10){\makebox(0,0){$\bullet$}} \put(70,20){\makebox(0,0){$\bullet$}}
\put(100,20){\makebox(0,0){$\bullet$}}
\put(10,15){\oval(18,28)} \put(40,15){\oval(18,28)} \put(70,15){\oval(18,28)} \put(100,20){\oval(18,18)}
\dashline{4}(10,20)(40,20) \dashline{4}(40,10)(70,10) \dashline{4}(70,20)(100,20)
\put(10,10){\line(3,1){30}} \put(40,20){\line(3,-1){30}} \put(70,10){\line(3,1){30}}
\put(40,10){\line(3,1){30}} \put(10,10){\line(1,0){30}}

\put(55,0){\makebox(0,0){$R6$}}
\end{picture}}

\put(0,50){\begin{picture}(110,40)(0,0)
\put(10,20){\makebox(0,0){$\bullet$}}
\put(40,10){\makebox(0,0){$\bullet$}} \put(40,20){\makebox(0,0){$\bullet$}}  \put(40,30){\makebox(0,0){$\bullet$}}
\put(70,10){\makebox(0,0){$\bullet$}} \put(70,20){\makebox(0,0){$\bullet$}}
\put(100,20){\makebox(0,0){$\bullet$}}
\put(10,20){\oval(18,18)} \put(40,20){\oval(18,38)} \put(70,15){\oval(18,28)} \put(100,20){\oval(18,18)}
\dashline{4}(10,20)(40,20) \dashline{4}(40,10)(70,10) \dashline{4}(70,20)(100,20)
\put(10,20){\line(3,1){30}} \put(40,30){\line(3,-2){30}}
\put(40,20){\line(3,-1){30}} \put(40,10){\line(3,1){30}} \put(70,10){\line(3,1){30}}

\put(55,0){\makebox(0,0){$R7$}}
\end{picture}}

\put(130,50){\begin{picture}(110,40)(0,0)
\put(10,20){\makebox(0,0){$\bullet$}}
\put(40,10){\makebox(0,0){$\bullet$}} \put(40,20){\makebox(0,0){$\bullet$}}  \put(40,30){\makebox(0,0){$\bullet$}}
\put(70,10){\makebox(0,0){$\bullet$}} \put(70,20){\makebox(0,0){$\bullet$}}
\put(100,20){\makebox(0,0){$\bullet$}}
\put(10,20){\oval(18,18)} \put(40,20){\oval(18,38)} \put(70,15){\oval(18,28)} \put(100,20){\oval(18,18)}
\dashline{4}(10,20)(40,20) \dashline{4}(40,10)(70,10) \dashline{4}(70,20)(100,20)
\put(10,20){\line(3,1){30}} \put(10,20){\line(3,-1){30}} \put(40,30){\line(3,-2){30}}
\put(40,10){\line(3,1){30}} \put(70,10){\line(3,1){30}}

\put(55,0){\makebox(0,0){$R8$}}
\end{picture}}

\put(0,0){\begin{picture}(110,40)(0,0)
\put(10,20){\makebox(0,0){$\bullet$}}
\put(40,10){\makebox(0,0){$\bullet$}} \put(40,20){\makebox(0,0){$\bullet$}}  \put(40,30){\makebox(0,0){$\bullet$}}
\put(70,10){\makebox(0,0){$\bullet$}} \put(70,20){\makebox(0,0){$\bullet$}}
\put(100,20){\makebox(0,0){$\bullet$}}
\put(10,20){\oval(18,18)} \put(40,20){\oval(18,38)} \put(70,15){\oval(18,28)} \put(100,20){\oval(18,18)}
\dashline{4}(10,20)(40,20) \dashline{4}(40,10)(70,10) \dashline{4}(70,20)(100,20)
\put(10,20){\line(3,1){30}} \put(10,20){\line(3,-1){30}} \put(40,30){\line(3,-2){30}}
\put(40,20){\line(1,0){30}} \put(70,10){\line(3,1){30}}

\put(55,0){\makebox(0,0){$R9$}}
\end{picture}}

\put(130,0){\begin{picture}(110,40)(0,0)
\put(10,20){\makebox(0,0){$\bullet$}}
\put(40,10){\makebox(0,0){$\bullet$}} \put(40,20){\makebox(0,0){$\bullet$}}
\put(70,10){\makebox(0,0){$\bullet$}} \put(70,20){\makebox(0,0){$\bullet$}}
\put(100,20){\makebox(0,0){$\bullet$}}
\put(10,20){\oval(18,18)} \put(40,15){\oval(18,28)} \put(70,15){\oval(18,28)} \put(100,20){\oval(18,18)}
\dashline{4}(10,20)(40,20) \dashline{4}(40,10)(70,10) \dashline{4}(70,20)(100,20)
\put(10,20){\line(3,-1){30}} \put(40,20){\line(1,0){30}} \put(70,10){\line(3,1){30}}
\put(40,10){\line(3,1){30}}

\put(55,0){\makebox(0,0){$R10$}}
\end{picture}}

\end{picture}
\caption{All degree-2 irreducible monotone patterns solved by SAC must
occur in at least one of these patterns.} \label{fig:open2}
\end{figure}

\thicklines \setlength{\unitlength}{2.5pt}
\begin{figure}
\centering
\begin{picture}(110,40)(0,0)
\put(10,20){\makebox(0,0){$\bullet$}}
\put(40,10){\makebox(0,0){$\bullet$}}
\put(40,20){\makebox(0,0){$\bullet$}}
\put(70,10){\makebox(0,0){$\bullet$}}
\put(70,20){\makebox(0,0){$\bullet$}}
\put(70,30){\makebox(0,0){$\bullet$}}
\put(100,20){\makebox(0,0){$\bullet$}}

\put(10,20){\oval(18,18)}
\put(40,15){\oval(18,28)}
\put(70,20){\oval(18,38)}
\put(100,20){\oval(18,18)}

\dashline{4}(10,20)(40,10)
\put(10,20){\line(1,0){30}}
\dashline{4}(40,20)(70,10)
\put(40,20){\line(1,0){30}}
\put(40,20){\line(3,1){30}}

\dashline{4}(70,20)(100,20)
\put(70,30){\line(3,-1){30}}
\end{picture}
\caption{The pattern \PatternRSevenMinus, a subpattern of \PatternRSeven.} \label{fig:R7}
\end{figure}
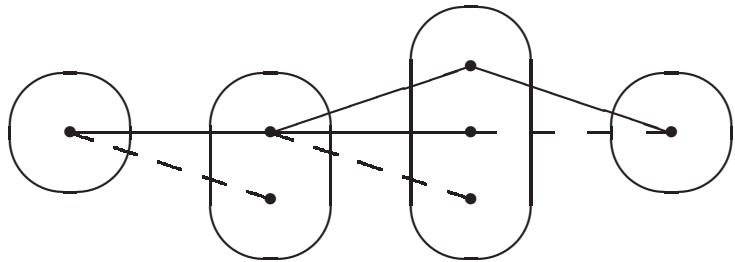

Consider the monotone patterns \PatternPThree\ and \PatternQThree\ shown
in~Figure~\ref{fig:open3}, patterns \PatternRFive, \PatternREight\ shown
in~Figure~\ref{fig:open2}, and pattern \PatternRSevenMinus\ shown
in~Figure~\ref{fig:R7}.

\begin{theorem*}[Main]
\label{thm:mainresult}
The patterns \PatternPThree, \PatternQThree, \PatternRFive, \PatternREight, and \PatternRSevenMinus\ are SAC-solvable.
\end{theorem*}

In order to prove the SAC-solvability of \PatternPThree, \PatternREight\ and
\PatternRSevenMinus\ we use the same idea of following the trace of arc
consistency and argue that the resulting instance is not too complicated. While
the same idea is behind the proofs of all three patterns, the technical details
differ.

In the remaining two cases we identify an operation that preserves
SAC and satisfiability, does not introduce the pattern and after repeated
application necessarily produces an equivalent instance which is solved by SAC. In the case of \PatternRFive, the operation is simply removing any 
constraint. In the case of \PatternQThree, the operation is
BTP-merging~\cite{cooper16:ai}.

\begin{remark} 
By Proposition~\ref{prop:app2} from Section~\ref{sec:CounterExamples},
any \emph{monotone} and \emph{irreducible} pattern solvable by SAC must occur in at least one of
the patterns shown in~Figures~\ref{fig:open3} and~\ref{fig:open2}. By this analysis, we have managed to reduce the number of remaining cases to a
handful. Our main result shows that some of these are SAC-solvable. In
particular, the patterns \PatternPThree, \PatternQThree, \PatternRFive, and
\PatternREight\ are maximal in the sense that adding anything to them would give
a pattern that is either non-monotone or not solved by SAC. 
\end{remark}

\begin{remark}
We point out that certain interesting forbidden patterns, such as
BTP~\cite{cooper10:ai-btp}, NegTrans~\cite{cz11:ai}, and
EMC~\cite{cz17:lmcs} are not monotone. On the other hand, the patterns
\PatternTOne$,\ldots,$\PatternTFive\ shown in Figure~\ref{figT}
\emph{are} monotone. Patterns \PatternTOne$,\ldots,$\PatternTFive\
were identified in~\cite{Cooper15:dam} as the maximal 
irreducible tractable patterns on two connected constraints. We show
in Section~\ref{sec:CounterExamples} that \PatternTOne\ is not
solved by SAC. Our main result implies (since
\PatternREight\ contains \PatternTFour\ and \PatternTFive) that both
\PatternTFour\ and \PatternTFive\ are solved by SAC. It can easily be shown,
from Lemma~\ref{lem:V-} and~\cite[Lemma~25]{Cooper15:dam}, that \PatternTTwo\ is solved by SAC,
and we provide, in Appendix~\ref{sec:T3}, a simple proof that T3 is
solved by SAC as well. Hence, we have characterised all 2-constraint irreducible patterns solvable by SAC.
\end{remark}

\begin{remark}
Observe that \PatternPThree\ does not occur in any binary CSP instance in which
all degree 3 or more variables are Boolean. This shows that 2-SAT is a strict subset of \CSP{\PatternPThree}. This class
is incomparable with language-based generalisations of 2-SAT, such as the class
ZOA~\cite{Cooper94:characterising} (the language of ``zero-one-all'' relations, that is, of all relations that admit the dual discriminator polymorphism) since in \CSP{\PatternPThree} degree-2
variables can be constrained by arbitrary constraints. Indeed, instances in
\CSP{\PatternPThree} can have an arbitrary constraint on the pair of variables
$x,y$, where $x$ is of arbitrary degree and of arbitrary domain size if for all
variables $z \notin \{x,y\}$, the constraint on the pair of variables $x,z$ is
of the form $(x \in S) \vee (z \in T_z)$ where $S$ is fixed (i.e. independent of
$z$) but $T_z$ is arbitrary. \PatternREight\ and \PatternRSevenMinus\ generalise \PatternTFour\
and \CSP{\PatternTFour} generalises ZOA~\cite{Cooper15:dam}, so
\CSP{\PatternREight} and \CSP{\PatternRSevenMinus} are strict generalisations
of ZOA.
\end{remark}

\thinlines \setlength{\unitlength}{1.1pt}
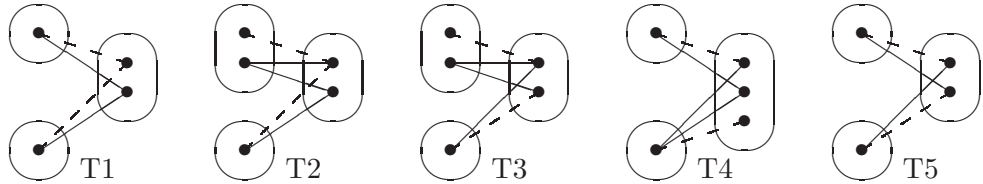
\begin{figure}[h]
\centering
\begin{picture}(332,63)(-2,-3)

\put(0,0){\begin{picture}(50,60)(0,0)
\put(10,10){\makebox(0,0){$\bullet$}} \dashline{4}(10,50)(40,40)
\put(10,50){\makebox(0,0){$\bullet$}} \dashline{4}(10,10)(40,40)
\put(40,30){\makebox(0,0){$\bullet$}} \put(10,10){\line(3,2){30}}
\put(40,40){\makebox(0,0){$\bullet$}} \put(10,50){\line(3,-2){30}}
\put(30,4){\makebox(0,0){\PatternTOne}} \put(10,10){\oval(20,20)}
\put(10,50){\oval(20,20)} \put(40,35){\oval(20,30)}
\end{picture}}

\put(70,0){\begin{picture}(50,60)(0,0)
\put(10,10){\makebox(0,0){$\bullet$}} \dashline{4}(10,50)(40,40)
\put(10,50){\makebox(0,0){$\bullet$}} \dashline{4}(10,10)(40,40)
\put(40,30){\makebox(0,0){$\bullet$}} \put(10,10){\line(3,2){30}}
\put(40,40){\makebox(0,0){$\bullet$}} \put(10,40){\line(3,-1){30}}
\put(10,40){\makebox(0,0){$\bullet$}} \put(10,40){\line(1,0){30}}
\put(30,4){\makebox(0,0){\PatternTTwo}} \put(10,10){\oval(20,20)}
\put(10,45){\oval(20,30)} \put(40,35){\oval(20,30)}
\end{picture}}

\put(140,0){\begin{picture}(50,60)(0,0)
\put(10,10){\makebox(0,0){$\bullet$}} \dashline{4}(10,50)(40,40)
\put(10,50){\makebox(0,0){$\bullet$}} \dashline{4}(10,10)(40,30)
\put(40,30){\makebox(0,0){$\bullet$}} \put(10,10){\line(1,1){30}}
\put(40,40){\makebox(0,0){$\bullet$}} \put(10,40){\line(3,-1){30}}
\put(10,40){\makebox(0,0){$\bullet$}} \put(10,40){\line(1,0){30}}
\put(30,4){\makebox(0,0){\PatternTThree}} \put(10,10){\oval(20,20)}
\put(10,45){\oval(20,30)} \put(40,35){\oval(20,30)}
\end{picture}}

\put(210,0){\begin{picture}(50,60)(0,0)
\put(10,10){\makebox(0,0){$\bullet$}} \dashline{4}(10,50)(40,40)
\put(10,50){\makebox(0,0){$\bullet$}} \dashline{4}(10,10)(40,20)
\put(40,30){\makebox(0,0){$\bullet$}} \put(10,10){\line(3,2){30}}
\put(40,40){\makebox(0,0){$\bullet$}} \put(10,50){\line(3,-2){30}}
\put(40,20){\makebox(0,0){$\bullet$}} \put(10,10){\line(1,1){30}}
\put(30,4){\makebox(0,0){\PatternTFour}} \put(10,10){\oval(20,20)}
\put(10,50){\oval(20,20)} \put(40,30){\oval(20,40)}
\end{picture}}

\put(280,0){\begin{picture}(50,60)(0,0)
\put(10,10){\makebox(0,0){$\bullet$}} \dashline{4}(10,50)(40,40)
\put(10,50){\makebox(0,0){$\bullet$}} \dashline{4}(10,10)(40,30)
\put(40,30){\makebox(0,0){$\bullet$}} \put(10,10){\line(1,1){30}}
\put(40,40){\makebox(0,0){$\bullet$}} \put(10,50){\line(3,-2){30}}
\put(30,4){\makebox(0,0){\PatternTFive}} \put(10,10){\oval(20,20)}
\put(10,50){\oval(20,20)} \put(40,35){\oval(20,30)}
\end{picture}}

\end{picture}
\caption{The set of tractable 2-constraint irreducible patterns.} \label{figT}
\end{figure}

\section{Notation for the Trace Technique}
\label{sec:trace}

Given a singleton arc consistent instance $I$, a variable $x$ and a
value $v \in D(x)$, we denote by $\ixv$ the instance obtained by
assigning $x$ to $v$ (that is, setting $D(x) = \{v\}$) and enforcing
arc consistency. To avoid confusion with the original domains, we
will use $\dxv{y}$ to denote the domain of the variable $y$ in
$\ixv$. For our proofs we will assume that arc consistency has been
enforced using a straightforward algorithm that examines the
constraints one at a time and removes the points that do not have a
support until a fixed point is reached. We will be interested in the
\textit{trace} of this algorithm, given as a chain of propagations:
$$\pxv : (x \rightarrow y_0), (x_1 \rightarrow y_1), (x_2 \rightarrow y_2), \ldots, (x_p \rightarrow y_p)$$
where $x_i \rightarrow y_i$ means that the algorithm has inferred a
change in the domain of $y_i$ when examining the constraint
$R(x_i,y_i)$. We define a map $\rho : \pxv \mapsto 2^D$ that maps
each $(x_i \rightarrow y_i) \in \pxv$ to the set of values that were
removed from $\dxv{y_i}$ at this step. Without loss of generality, we
assume that the steps $(x_i \rightarrow y_i)$ such that the pruning
of $\rho(x_i \rightarrow y_i)$ from $\dxv{y_i}$ does not incur further
propagation are performed last.

We denote by $\spxv$ the set of variables that appear in $\pxv$.
Because $I$ was (singleton) arc consistent before $x$ was assigned,
we have $\spxv = \{x\} \cup \{y_i \mid i \geq 0 \}$. We rename the
elements of $\spxv$ as $\{p_i \mid i \geq 0\}$ where the index $i$
denotes the order of first appearance in $\pxv$. Finally, we use
$\sipxv$ to denote the set of \textit{inner} variables, that is, the
set of all variables $p_j \in \spxv$ for which there exists $p_r \in
\spxv$ such that $(p_j \rightarrow p_r) \in \pxv$.

\section{Tractability of \PatternPThree}
\label{sec:PatternPThree}

Consider the pattern \PatternPThree\ shown in
Figure~\ref{fig:open3}. Let $I \in$ \CSP{\PatternPThree} be a
singleton arc consistent instance, $x$ be any variable and $v$ be any
value in the domain of $x$. Our proof of the SAC-decidability of
\CSP{\PatternPThree} uses the trace of the arc consistency algorithm
to determine a subset of variables in the vicinity of $x$ such that
$(i)$ the projection of $\ixv$ to this particular subset is satisfiable, $(ii)$ those variables do not interact too much with the
rest of the instance and $(iii)$ the projections of $\ixv$ and $I$ on
the other variables are almost the same. We then use these three
properties to show that the satisfiability of $I$ is equivalent to
that of an instance with fewer variables, and we repeat the operation
until the smaller instance is trivially satisfiable.

The following lemma describes the particular structure of $\ixv$
around the variables whose domain has been reduced by arc
consistency. Note that a non-trivial constraint in $I$ can be trivial in $I_{xv}$ because of domain changes; unless otherwise stated the triviality/non-triviality of constraints is always discussed with respect to $I_{xv}$.

\begin{lemma}
\label{lem:paths} Consider the instance $\ixv$. Every variable $p_i
\in \sipxv$ is in the scope of at most two non-trivial constraints,
which must be of the form $R(p_j,p_i)$ and $R(p_i,p_r)$ with $j < i$,
$(p_j \rightarrow p_i) \in \pxv$ and $(p_i \rightarrow p_r) \in
\pxv$.
\end{lemma}

\begin{proof}
The claim is true for $p_0 = x$ as every constraint incident to $x$
is trivial. Otherwise, let $p_i \in \sipxv$ be such that $p_i \neq
x$. Let $p_j$, $j < i$ be such that $(p_j \rightarrow p_i)$ occurs
first in $\pxv$. Because $p_i \in \sipxv$ and we assumed that the arc
consistency algorithm performs the pruning that do not incur further
propagation last, we know that there exists $c_i \in \rho(p_j
\rightarrow p_i)$ and $p_r \in \spxv$ with $(p_i \rightarrow p_r) \in
\pxv$ such that the pruning of $c_i$ from $D(p_i)$ allows the
pruning of some $a_r \in \rho(p_i \rightarrow p_r)$ from the domain
of $p_r$. It follows that $(c_i,a_r) \in
R(p_i,p_r)$, $(v_i,a_r) \notin R(p_i,p_r)$ for any $v_i \in
\dxv{p_i}$ and $(v_j,c_i) \notin R(p_j,p_i)$ for any $v_j \in
\dxv{p_j}$. Moreover, $a_r$ was a support for $c_i$ at $p_r$ when $c_i$ was pruned so we know that $p_j \neq p_r$.

For the sake of contradiction, let us assume that there exists a
constraint $R(p_i,l)$ with $l \notin \{p_j,p_r\}$ that is not
trivial. In particular, there exist $a_i,b_i \in \dxv{p_i}$ and $a_l
\in \dxv{l}$ such that $(a_i,a_l) \in R(p_i,l)$ but $(b_i,a_l) \notin
R(p_i,l)$. Since $a_i$ is in $\dxv{p_i}$ and $a_r$ was removed by arc consistency when
inspecting the constraint $R(p_i,p_r)$, we have $(a_i,a_r) \notin
R(p_i,p_r)$. $\ixv$ is arc consistent so there exists some $a_j \in
\dxv{p_j}$ such that $(a_j,b_i) \in R(p_j,p_i)$, and since $c_i \in
\rho(p_j \rightarrow p_i)$ we have $(a_j,c_i) \notin R(p_j,p_i)$. At
this point we have reached the desired contradiction as
\PatternPThree\ occurs on $(p_i,p_j,p_r,l)$ with $p_i$ being the
middle variable (see Figure~\ref{fig:p3paths}).
\end{proof}

\thicklines \setlength{\unitlength}{1.4pt}
\begin{figure}
\centering
\begin{picture}(100,100)(0,-5)
\put(10,10){\makebox(0,0){$\bullet$}}
\put(7,15){\makebox(0,0){$a_j$}}
\put(10,10){\oval(18,22)}
\put(10,-6){\makebox(0,0){$p_j$}}
\put(10,90){\makebox(0,0){$\bullet$}}
\put(10,95){\makebox(0,0){$a_r$}}
\put(10,90){\oval(18,22)}
\put(10,74){\makebox(0,0){$p_r$}}
\put(90,50){\makebox(0,0){$\bullet$}}
\put(90,55){\makebox(0,0){$a_l$}}
\put(90,50){\oval(18,22)}
\put(90,34){\makebox(0,0){$l$}}
\put(40,40){\makebox(0,0){$\bullet$}}
\put(42,44.5){\makebox(0,0){$b_i$}}
\put(40,50){\makebox(0,0){$\bullet$}}
\put(42,54){\makebox(0,0){$a_i$}}
\put(40,60){\makebox(0,0){$\bullet$}}
\put(42,64){\makebox(0,0){$c_i$}}
\put(40,50){\oval(18,38)}
\put(40,25){\makebox(0,0){$p_i$}}

\dashline{3}(10,10)(40,60)
\dashline{3}(40,40)(90,50)
\dashline{3}(10,90)(40,50)
\put(10,10){\line(1,1){30}}
\put(10,90){\line(1,-1){30}}
\put(40,50){\line(1,0){50}}
\end{picture}
\caption{The occurence of \PatternPThree\ in the proof of Lemma~\ref{lem:paths}.}
\label{fig:p3paths}
\end{figure}
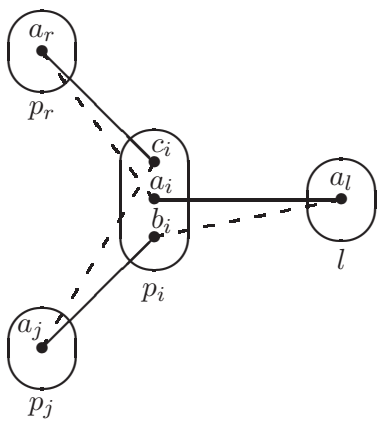

Given a subset $S$ of variables, an \textit{$S$-path} between two variables $y_1$ and $y_2$ is a path
$R(y_1,x_2),R(x_2,x_3),\ldots,R(x_{k},y_2)$ of non-trivial
constraints with $k \geq 2$ and $x_2,\ldots,x_{k} \in S$.

\begin{lemma}
\label{lem:xpath} Consider the instance $\ixv$. There is no
$(\sipxv)$-path between two variables in $\mvar \backslash \sipxv$
and there is no cycle of non-trivial constraints in $\ixv[\sipxv]$.
\end{lemma}

\begin{proof}
Let $y_1,y_2 \in \mvar \backslash \sipxv$ and assume for the sake
of contradiction that a $(\sipxv)$-path
$R(y_1,x_2),R(x_2,x_3),\ldots,R(x_{k-1},y_2)$ exists. Let $p_i \in
\{x_2,\ldots,x_{k-1}\}$ be such that $i$ is minimum.
Since $p_i$ is in the scope of two non-trivial constraints in this path,
it follows from Lemma~\ref{lem:paths} that $p_i$ is in the scope
of exactly two non-trivial constraints, one of which is
of the form $R(p_j,p_i)$ with $j < i$ and $(p_j \rightarrow p_i) \in
\pxv$. It follows from $(p_j \rightarrow p_i) \in
\pxv$ that $p_j \in \sipxv$ and
hence $p_j$ is not an endpoint of the path, and then $j < i$  contradicts the minimality of $i$. The second part of the claim follows from the same argument, by
considering a cycle as a $(\sipxv)$-path
$R(x_1,x_2),R(x_2,x_3),\ldots,R(x_{k-1},x_1)$ with $x_1 \in (\sipxv)$
and defining $p_i$ as the variable among $\{x_1,\ldots,x_{k-1}\}$
with minimum index.
\end{proof}

\begin{lemma}
\label{lem:solfix}
$\ixv$ has a solution if and only if $\ixv[\mvar \backslash \sipxv]$ has a solution.
\end{lemma}

\begin{proof}
The ``only if'' implication is trivial, so we focus on the other
direction. Suppose that there exists a solution $\phi$ to $\ixv[\mvar
\backslash \sipxv]$. Let $Y$ be a set of variables initialized to
$\mvar \backslash \sipxv$. We will grow $Y$ with the invariants that
$(i)$ we know a solution $\phi$ to $\ixv[Y]$, and $(ii)$ there is no
$(\mvar \backslash Y)$-path between two variables in $Y$ (which is
true at the initial state by Lemma~\ref{lem:xpath}).

If there is no non-trivial constraint between $\mvar \backslash Y$
and $Y$ then $\ixv$ is satisfiable if and only if $\ixv[\mvar
\backslash Y]$ is. By construction $\mvar \backslash Y \subseteq
\sipxv$ and by Lemma~\ref{lem:xpath} we know that $\ixv[\mvar
\backslash Y]$ has no cycle of non-trivial constraints. Because
$\ixv[\mvar \backslash Y]$ is arc consistent and acyclic it has a
solution~\cite{Freuder82:acm}, and we can conclude that in this case $\ixv$ has a
solution.

Otherwise, let $p_i \in \mvar \backslash Y$ be such that there exists
a non-trivial constraint between $p_i$ and some variable $y \in Y$.
By $(ii)$, this non-trivial constraint must be unique (with respect to $p_i$) as otherwise we
would have a $(\mvar \backslash Y)$-path between two variables in
$Y$. By arc consistency, there exists $a_i \in \dxv{p_i}$ such that
$(a_i,\phi(y)) \in R(p_i,y)$; because this non-trivial constraint
is unique, setting $\phi(p_i) = a_i$ yields a solution to $\ixv[Y
\cup \{p_i\}]$. Because any $(\mvar \backslash (Y \cup
\{p_i\}))$-path between two variables in $Y \cup \{p_i\}$ would
extend to a $(\mvar \backslash Y)$-path between $Y$
variables by going through $p_i$, we know that no such path exists.
Then $Y \gets Y \cup \{p_i\}$ satisfies both invariants, so we can
repeat the operation until we have a solution to the whole instance
or all constraints between $Y$ and $\mvar \backslash Y$ are trivial.
In both cases $\ixv$ has a solution.
\end{proof}

\begin{lemma}
\label{lem:outer}
$I$ has a solution if and only if $I[\mvar \backslash \sipxv]$ has a solution.
\end{lemma}

\begin{proof}
Again the ``only if'' implication is trivial so we focus on the other
direction. Let us assume for the sake of contradiction that $I[\mvar
\backslash \sipxv]$ has a solution but $I$ does not. In
particular  this implies that $\ixv$ does not have a solution, and
then by Lemma~\ref{lem:solfix} we know that $\ixv[\mvar \backslash
\sipxv]$ has no solution either. We define $Z$ as a subset of $\mvar \backslash
\sipxv$ of minimum size such that $\ixv[Z]$ has no solution. Observe that $\ixv[Z]$ can only differ from
$I[Z]$ by having fewer values in the domain of
the variables in $\spxv$. Let $\phi$ be a solution to $I[Z]$ such that $\phi(y) \in \dxv{y}$ for as many variables $y$ as possible. Because $\phi$ is not a solution to $\ixv[Z]$, there exists $p_r \in
Z \cap \spxv$ and $p_j \in \sipxv$ such that $(p_j
\rightarrow p_r) \in \pxv$ and $\phi(p_r) \in \rho(p_j \rightarrow
p_r)$ (recall that $\rho(p_j \rightarrow p_r)$ is the set of points removed by the AC algorithm in the domain of $p_r$ at step $(p_j \rightarrow p_r)$). By construction, $p_j \notin Z$.

First, let us assume that there exists a variable $y \in Z$, $y \neq p_r$ such that there is no $a_r \in \dxv{p_r}$ with $(\phi(y),a_r) \in R(y,p_r)$. This implies, in particular, that $\phi(y) \notin \dxv{y}$. We first prove that $R(y,p_r)$ and $R(p_j,p_r)$ are the only possible non-trivial constraints involving $p_r$ in $\ixv$. If there exists a fourth variable $z$ such that $R(p_r,z)$ is non-trivial in $\ixv$, then there exist $a_r,b_r \in \dxv{p_r}$ and $a_z \in \dxv{z}$ such that $(a_r,a_z) \in R(p_r,z)$ but $(b_r,a_z) \notin R(p_r,z)$. By assumption we have $(\phi(y),a_r) \notin R(y,p_r)$ and $(\phi(y),\phi(p_r)) \in R(y,p_r)$. Finally, $b_r$ has a support $a_j \in \dxv{p_j}$ and $\phi(p_r) \in \rho(p_j \rightarrow p_r)$ so we have $(a_j,a_r) \in R(p_j,p_r)$ but $(a_j,\phi(p_r)) \notin R(p_j,p_r)$. This produces \PatternPThree\ on $(p_r,y,p_j,z)$ with $p_r$ being the middle variable. Therefore, we know that $R(y,p_r)$ and $R(p_j,p_r)$ are the only possible non-trivial constraints involving $p_r$ in $\ixv$. However, in this case the variable $p_r$ has only one incident non-trivial constraint in $\ixv[Z]$, and hence $\ixv[Z]$ has a solution if and only if $\ixv[Z \backslash p_r]$ has one. This contradicts the minimality of $Z$, and for the rest of the proof we can assume that for every $y \in Z$ there exists some $a_r \neq \phi(p_r)$ such that $a_r \in \dxv{p_r}$ and $(\phi(y),a_r) \in R(y,p_r)$.

Now, let $y \in Z$ be such that $y \neq p_r$ and $|\{ b \in \dxv{p_r}
\mid (\phi(y),b) \in R(y,p_r) \}|$ is minimum. By the argument above, there exists $a_r \in \dxv{p_r}$ such that $(\phi(y),a_r) \in R(y,p_r)$ and $a_r \neq \phi(p_r)$. By the choice of $\phi$ and $p_r$, setting $\phi(p_r)
= a_r$ would violate at least one constraint in $I[Z]$, so there exists some variable $z \in Z$, $z \neq y$
such that $(\phi(z),a_r) \notin R(z,p_r)$. Furthermore, by arc
consistency of $\ixv$ there exists $a_j \in \dxv{p_j}$ such that
$(a_j,a_r) \in R(p_j,p_r)$. Recall that we picked $p_j$ in such a way
that $\phi(p_r) \in \rho(p_j \rightarrow p_r)$, and so we have
$(a_j,\phi(p_r)) \notin R(p_j,p_r)$. We summarize what we have in Figure~\ref{fig:lemsol}.
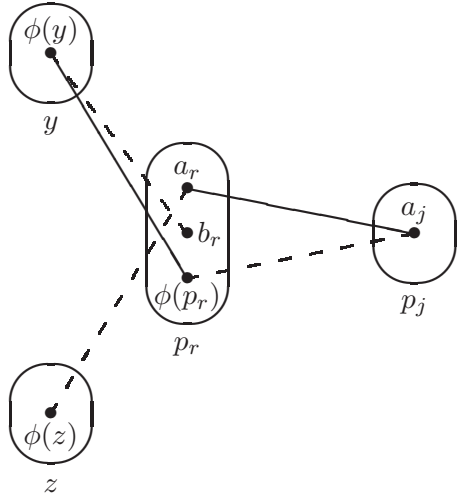
\begin{figure}
\centering
\thicklines \setlength{\unitlength}{1.7pt}
\begin{picture}(100,100)(0,-5)
\put(10,10){\makebox(0,0){$\bullet$}}
\put(10,5){\makebox(0,0){$\phi(z)$}}
\put(10,10){\oval(18,22)}
\put(10,-6){\makebox(0,0){$z$}}
\put(10,90){\makebox(0,0){$\bullet$}}
\put(10,95){\makebox(0,0){$\phi(y)$}}
\put(10,90){\oval(18,22)}
\put(10,74){\makebox(0,0){$y$}}
\put(90,50){\makebox(0,0){$\bullet$}}
\put(90,54.5){\makebox(0,0){$a_j$}}
\put(90,50){\oval(18,22)}
\put(90,34){\makebox(0,0){$p_j$}}
\put(40,40){\makebox(0,0){$\bullet$}}
\put(40,36){\makebox(0,0){$\phi(p_r)$}}
\put(40,50){\makebox(0,0){$\bullet$}}
\put(45,50){\makebox(0,0){$b_r$}}
\put(40,60){\makebox(0,0){$\bullet$}}
\put(40,64){\makebox(0,0){$a_r$}}
\put(40,50){\oval(18,40)}
\put(40,25){\makebox(0,0){$p_r$}}

\dashline{3}(10,10)(40,60)
\dashline{3}(40,40)(90,50)
\dashline{3}(10,90)(40,50)
\put(10,90){\line(3,-5){30}}
\put(40,60){\line(5,-1){50}}
\end{picture}
\caption{Some positive and negative edges between $y$, $z$, $p_j$ and $p_r$. The positive edges $\phi(y)a_r$ and $\phi(z)\phi(p_r)$ are omitted for clarity; $b_r$ is any value in $\dxv{p_r}$ that is not compatible with $\phi(y)$.}
\label{fig:lemsol}
\end{figure}
Observe that unless \PatternPThree\ occurs, for every $b_r \in
\dxv{p_r}$ such that $(\phi(y),b_r) \notin R(y,p_r)$ we also have
$(\phi(z),b_r) \notin R(z,p_r)$. However, recall that $(\phi(y),a_r)
\in R(y,p_r)$ so $\phi(z)$ is compatible with strictly fewer values
in $\dxv{p_r}$ than $\phi(y)$. This contradicts the choice of $y$. It
follows that setting $\phi(p_r) = a_r$ cannot violate any constraint
in $I[Z]$, which is impossible by our choice of
$\phi$ - a final contradiction.
\end{proof}

\begin{theorem}
\CSP{\PatternPThree} is solved by singleton arc consistency.
\end{theorem}

\begin{proof}
Let $I \in$ \CSP{\PatternPThree} be singleton arc consistent. Pick
any variable $x$ and value $v \in D(x)$. By singleton arc consistency
the instance $\ixv$ does not have any empty domains. If $\sipxv$ is empty then $I$ has a solution if and only if $I[\mvar \backslash \{x\}]$ has one. Otherwise, by Lemma~\ref{lem:outer}, $I$ has a solution if and only if $I[\mvar
\backslash \sipxv]$ has one. In the latter case we must have $x \in \sipxv$, so overall we can conclude that $I$ has a solution if and only if $I[\mvar
\backslash (\sipxv \cup \{x\})]$ has one. Because $I[\mvar \backslash (\sipxv \cup \{x\})]$ is
singleton arc consistent as well and $\sipxv \cup \{x\} \neq \emptyset$ we
can repeat the procedure until $\mvar \backslash (\sipxv \cup \{x\})$ is empty, at
which point we may conclude that $I$ has a solution.
\end{proof}

\section{Tractability of \PatternREight\ and \PatternRSevenMinus}
\label{sec:R87-}

\PatternPThree\ and \PatternREight\ (Figure~\ref{fig:open2}) are structurally
dissimilar, but the idea of using $\ixv$ and the trace of the arc consistency algorithm to extract
variables from $I$ without altering satisfiability works in the case
of \PatternREight\ as well. We define a \textit{star} to be a non-empty set of constraints whose scopes all intersect.   The \textit{centers} of a star are its variables of highest degree (every
star with three or more variables has a unique center).
The following lemma is the \PatternREight\ analog of Lemma~\ref{lem:paths}; the main differences are a slightly stronger prerequisite (no neighbourhood substitutable values) and that arc consistency leaves stars of non-trivial constraints instead of paths.  

\begin{lemma}
\label{lem:path}
Let $I=(\mvar,\mdom,\mcons) \in$ \CSP{\PatternREight} be singleton arc consistent. Let $x \in \mvar$, $v \in D(x)$ and consider the instance $\ixv$. After the removal of
every neighbourhood substitutable value, every connected component of
non-trivial constraints that intersect with $\spxv$ is a star with a
center in $\spxv$.
\end{lemma}

\begin{proof}
We proceed by induction. Suppose that all neighbourhood substitutable values have been removed. First, no connected component of non-trivial
constraints may contain $p_0 = x$. Then, let $k \geq 0$ and
suppose that every connected component of non-trivial constraints
that intersect $\{p_i \mid i \leq k \}$ is a star centered on
$\spxv$. Suppose also, for the sake of contradiction, that there
exists a connected component $\mathcal{G}$ of non-trivial constraints
that contains $p_{k+1}$ and that is \textit{not} a star centered on
$\spxv$.

Let $p_j$, $j \leq k$ be such that $(p_j \rightarrow p_{k+1}) \in
\pxv$. By the induction hypothesis, $p_j$ cannot be part of
$\mathcal{G}$ and hence $R(p_j,p_{k+1})$ must be trivial.
Furthermore, if every simple path of non-trivial constraints starting
at $p_{k+1}$ had length 1 then $\mathcal{G}$ would be a star centered
on $p_{k+1}$, which would contradict our assumption. Therefore, there
exist two distinct variables $z_1,z_2 \notin \{p_j,p_{k+1}\}$ such
that neither $R(p_{k+1},z_1)$ nor $R(z_1,z_2)$ is trivial (again,
the claim $z_1,z_2 \neq p_j$ comes from the fact that $p_j$ is not
part of $\mathcal{G}$).

Because $R(p_{k+1},z_1)$ is not trivial, there exist two distinct
values $a_{k+1},b_{k+1} \in \dxv{p_{k+1}}$ and $a_1 \in \dxv{z_1}$
such that $(a_{k+1},a_1) \in R(p_{k+1},z_1)$ but $(b_{k+1},a_1)
\notin R(p_{k+1},z_1)$. Furthermore, $R(p_j,p_{k+1})$ is trivial and
hence there exists $a_j \in \dxv{p_j}$ such that
$(a_j,a_{k+1}),(a_j,b_{k+1}) \in R(p_j,p_{k+1})$. Finally, since
$(p_j \rightarrow p_{k+1}) \in \pxv$ some propagation must have taken
place in the domain of $p_{k+1}$, and hence there exists $c_{k+1}$
such that $(a_j,c_{k+1}) \notin R(p_j,p_{k+1})$. We can summarize
what we have in the following picture (the tuple $(a_1,a_2)$ comes
from the fact that $a_1$ must have a support in $R(z_1,z_2)$).

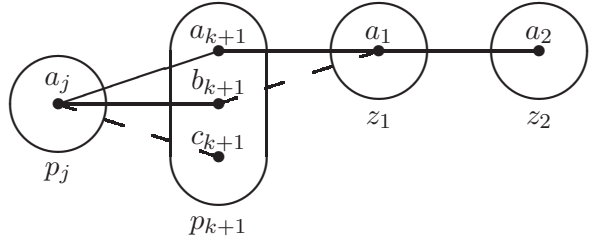
\thicklines \setlength{\unitlength}{2pt}
\begin{center}
\begin{picture}(110,40)(0,0)

\put(10,20){\makebox(0,0){$\bullet$}}
\put(10,24){\makebox(0,0){$a_j$}}

\put(40,10){\makebox(0,0){$\bullet$}}
\put(40,13){\makebox(0,0){$c_{k+1}$}}
\put(40,20){\makebox(0,0){$\bullet$}}
\put(40,24){\makebox(0,0){$b_{k+1}$}}
\put(40,30){\makebox(0,0){$\bullet$}}
\put(40,33){\makebox(0,0){$a_{k+1}$}}

\put(70,30){\makebox(0,0){$\bullet$}}
\put(70,33){\makebox(0,0){$a_1$}}

\put(100,30){\makebox(0,0){$\bullet$}}
\put(100,33){\makebox(0,0){$a_2$}}

\put(10,20){\oval(18,18)}
\put(10,7){\makebox(0,0){$p_j$}}
\put(40,20){\oval(18,38)}
\put(40,-3){\makebox(0,0){$p_{k+1}$}}
\put(70,30){\oval(18,18)}
\put(70,17){\makebox(0,0){$z_1$}}
\put(100,30){\oval(18,18)}
\put(100,17){\makebox(0,0){$z_2$}}

\put(10,20){\line(3,1){30}}
\put(10,20){\line(3,0){30}}

\put(40,30){\line(3,0){30}}

\put(70,30){\line(3,0){30}}

\dashline{4}(10,20)(40,10)
\dashline{4}(40,20)(70,30)
\end{picture}
\end{center}

\vspace{3mm}

\textbf{First case: there exists $\bm{b_1 \in \dxv{z_1}}$ such that
$\bm{(b_1,a_2) \notin R(z_1,z_2)}$}. For \PatternREight\ not to
occur, $(b_{k+1},b_1)$ must not belong to $R(p_{k+1},z_1)$. By arc
consistency, $b_1$ must be connected to some $d_{k+1} \in
\dxv{p_{k+1}}$. If there exists one such $d_{k+1}$ such that $(d_{k+1},a_1) \notin R(p_{k+1},z_1)$, then \PatternREight\ occurs
again, so $a_1$ dominates $b_1$ in the constraint $R(p_{k+1},z_1)$.
However, recall that all neighbourhood substitutable values have been
removed, so there must exist a variable $z_3$ (potentially equal to
$z_2$, but different from $p_j,p_{k+1},z_1$) and $a_3 \in \dxv{z_3}$
such that $(b_1,a_3) \in R(z_1,z_3)$ but $(a_1,a_3) \notin
R(z_1,z_3)$. Finally, because $b_{k+1}$ is arc consistent, there
exists $c_1 \in D(z_1)$ such that $(b_{k+1},c_1) \in R(p_{k+1},z_1)$.
We obtain the following two structures, which may only differ on the
last constraint.

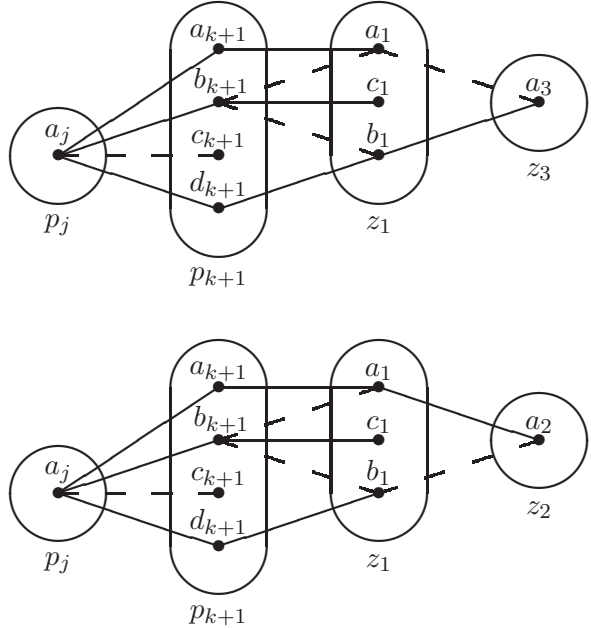
\thicklines \setlength{\unitlength}{2pt}
\begin{center}
\begin{picture}(110,50)(0,0)

\put(10,20){\makebox(0,0){$\bullet$}}
\put(10,24){\makebox(0,0){$a_j$}}

\put(40,10){\makebox(0,0){$\bullet$}}
\put(40,14.5){\makebox(0,0){$d_{k+1}$}}
\put(40,20){\makebox(0,0){$\bullet$}}
\put(40,23){\makebox(0,0){$c_{k+1}$}}
\put(40,30){\makebox(0,0){$\bullet$}}
\put(40.5,34){\makebox(0,0){$b_{k+1}$}}
\put(40,40){\makebox(0,0){$\bullet$}}
\put(40,43){\makebox(0,0){$a_{k+1}$}}

\put(70,20){\makebox(0,0){$\bullet$}}
\put(70,24){\makebox(0,0){$b_1$}}
\put(70,30){\makebox(0,0){$\bullet$}}
\put(70,33){\makebox(0,0){$c_1$}}
\put(70,40){\makebox(0,0){$\bullet$}}
\put(70,43){\makebox(0,0){$a_1$}}

\put(100,30){\makebox(0,0){$\bullet$}}
\put(100,33){\makebox(0,0){$a_3$}}

\put(10,20){\oval(18,18)}
\put(10,7){\makebox(0,0){$p_j$}}
\put(40,25){\oval(18,48)}
\put(40,-3){\makebox(0,0){$p_{k+1}$}}
\put(70,30){\oval(18,38)}
\put(70,7){\makebox(0,0){$z_1$}}
\put(100,30){\oval(18,18)}
\put(100,17){\makebox(0,0){$z_3$}}

\put(10,20){\line(3,1){30}}
\put(10,20){\line(3,2){30}}
\put(10,20){\line(3,-1){30}}

\put(40,40){\line(3,0){30}}
\put(40,30){\line(3,0){30}}
\put(40,10){\line(3,1){30}}

\put(70,20){\line(3,1){30}}

\dashline{4}(10,20)(40,20)
\dashline{4}(40,30)(70,40)
\dashline{4}(40,30)(70,20)
\dashline{4}(70,40)(100,30)
\end{picture}
\end{center}

\vspace{3mm}

\thicklines \setlength{\unitlength}{2pt}
\begin{center}
\begin{picture}(110,50)(0,0)

\put(10,20){\makebox(0,0){$\bullet$}}
\put(10,24){\makebox(0,0){$a_j$}}

\put(40,10){\makebox(0,0){$\bullet$}}
\put(40,14.5){\makebox(0,0){$d_{k+1}$}}
\put(40,20){\makebox(0,0){$\bullet$}}
\put(40,23){\makebox(0,0){$c_{k+1}$}}
\put(40,30){\makebox(0,0){$\bullet$}}
\put(40.5,34){\makebox(0,0){$b_{k+1}$}}
\put(40,40){\makebox(0,0){$\bullet$}}
\put(40,43){\makebox(0,0){$a_{k+1}$}}

\put(70,20){\makebox(0,0){$\bullet$}}
\put(70,24){\makebox(0,0){$b_1$}}
\put(70,30){\makebox(0,0){$\bullet$}}
\put(70,33){\makebox(0,0){$c_1$}}
\put(70,40){\makebox(0,0){$\bullet$}}
\put(70,43){\makebox(0,0){$a_1$}}

\put(100,30){\makebox(0,0){$\bullet$}}
\put(100,33){\makebox(0,0){$a_2$}}

\put(10,20){\oval(18,18)}
\put(10,7){\makebox(0,0){$p_j$}}
\put(40,25){\oval(18,48)}
\put(40,-3){\makebox(0,0){$p_{k+1}$}}
\put(70,30){\oval(18,38)}
\put(70,7){\makebox(0,0){$z_1$}}
\put(100,30){\oval(18,18)}
\put(100,17){\makebox(0,0){$z_2$}}

\put(10,20){\line(3,1){30}}
\put(10,20){\line(3,2){30}}
\put(10,20){\line(3,-1){30}}

\put(40,40){\line(3,0){30}}
\put(40,30){\line(3,0){30}}
\put(40,10){\line(3,1){30}}

\put(70,40){\line(3,-1){30}}

\dashline{4}(10,20)(40,20)
\dashline{4}(40,30)(70,40)
\dashline{4}(40,30)(70,20)
\dashline{4}(70,20)(100,30)
\end{picture}
\end{center}

\vspace{3mm}

The key observation here is that whenever $a_1$ or $b_1$ is compatible with any
value $v$ of a fourth variable $y \notin \{p_j, p_{k+1}, z_1\}$, then $c_1$ is compatible with $v$ as well unless
\PatternREight\ occurs. Thus, the only constraint on which $c_1$ may not
dominate both $a_1$ and $b_1$ is $R(p_{k+1},z_1)$. However, if $(d_{k+1},c_1)
\notin R(p_{k+1},z_1)$ then \PatternREight\ occurs in $(p_j,p_{k+1},z_1,z_2)$,
and if $(a_{k+1},c_1) \notin R(p_{k+1},z_1)$ then \PatternREight\ occurs in
$(p_j,p_{k+1},z_1,z_3)$; this is true for any choice of $a_{k+1}$ and
$d_{k+1}$ so $c_1$ dominates both $a_1$ and $b_1$ in $R(p_{k+1},z_1)$ - a contradiction,
since it means that $a_1$ and $b_1$ should have been removed by neighbourhood substitution.

\textbf{Second case: there does not exist $\bm{b_1 \in \dxv{z_1}}$
such that $\bm{(b_1,a_2) \notin R(z_1,z_2)}$ for any choice of
$\bm{z_2}$}. This means that we must have $(v_1,v_2) \in R(z_1,z_2)$
for all $v_1 \in \dxv{z_1}$ and $v_2 \in \dxv{z_2}$ such that
$(a_1,v_2) \in R(z_1,z_2)$. Putting this together with the fact that
by hypothesis $R(z_1,z_2)$ is not trivial, there exists $b_2 \in
\dxv{z_2}$ such that $(a_1,b_2) \notin R(z_1,z_2)$. Then $b_2$ must
have a support $(b_1,b_2)$ in $R(z_1,z_2)$, and $b_1$ must have a
support $(d_{k+1},b_1)$ in $R(p_{k+1},z_1)$. Because $R(p_j,p_{k+1})$
is trivial, $(a_j,d_{k+1}) \in R(p_j,p_{k+1})$. Let us update our
picture:

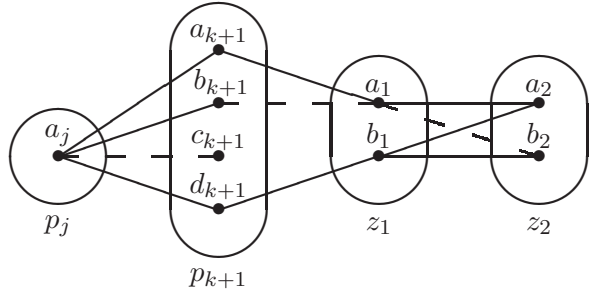
\thicklines \setlength{\unitlength}{2pt}
\begin{center}
\begin{picture}(110,50)(0,0)

\put(10,20){\makebox(0,0){$\bullet$}}
\put(10,24){\makebox(0,0){$a_j$}}

\put(40,10){\makebox(0,0){$\bullet$}}
\put(40,14.5){\makebox(0,0){$d_{k+1}$}}
\put(40,20){\makebox(0,0){$\bullet$}}
\put(40,23){\makebox(0,0){$c_{k+1}$}}
\put(40,30){\makebox(0,0){$\bullet$}}
\put(40.5,34){\makebox(0,0){$b_{k+1}$}}
\put(40,40){\makebox(0,0){$\bullet$}}
\put(40,43){\makebox(0,0){$a_{k+1}$}}

\put(70,20){\makebox(0,0){$\bullet$}}
\put(70,24){\makebox(0,0){$b_1$}}
\put(70,30){\makebox(0,0){$\bullet$}}
\put(70,33){\makebox(0,0){$a_1$}}

\put(100,20){\makebox(0,0){$\bullet$}}
\put(100,24){\makebox(0,0){$b_2$}}
\put(100,30){\makebox(0,0){$\bullet$}}
\put(100,33){\makebox(0,0){$a_2$}}

\put(10,20){\oval(18,18)}
\put(10,7){\makebox(0,0){$p_j$}}
\put(40,25){\oval(18,48)}
\put(40,-3){\makebox(0,0){$p_{k+1}$}}
\put(70,25){\oval(18,28)}
\put(70,7){\makebox(0,0){$z_1$}}
\put(100,25){\oval(18,28)}
\put(100,7){\makebox(0,0){$z_2$}}

\put(10,20){\line(3,1){30}}
\put(10,20){\line(3,2){30}}
\put(10,20){\line(3,-1){30}}

\put(40,40){\line(3,-1){30}}
\put(40,10){\line(3,1){30}}

\put(70,30){\line(3,0){30}}
\put(70,20){\line(3,0){30}}
\put(70,20){\line(3,1){30}}

\dashline{4}(10,20)(40,20)
\dashline{4}(40,30)(70,30)
\dashline{4}(70,30)(100,20)
\end{picture}
\end{center}

\vspace{3mm}

Observe that $a_{k+1}$ is an arbitrary value of $\dxv{p_{k+1}}$ that
is compatible with $a_1$. If $(a_{k+1},b_1) \notin R(p_{k+1},z_1)$,
then \PatternREight\ occurs. Hence, every value compatible with $a_1$
in \textit{every} constraint involving $z_1$ is also compatible with
$b_1$. This means that $a_1$ should have been removed by
neighbourhood substitution - a final contradiction.
\end{proof}

In the proof of SAC-solvability of \PatternPThree, only inner variables are extracted from the instance. The above lemma suggests that in the case of \PatternREight\ it is more convenient to extract all variables in $\spxv$, plus any variable that can be reached from those via a non-trivial constraint.

\begin{lemma}
\label{lem:sep}
Let $I=(\mvar,\mdom,\mcons) \in$ \CSP{\PatternREight} be singleton arc consistent. Let $x \in \mvar$, $v \in D(x)$ and consider the instance $\ixv$. After the removal of every neighbourhood substitutable value, there exists a partition $(X_1,X_2)$ of $\mvar$ such that
\begin{itemize}
\item $\spxv \subseteq X_1$;
\item $\forall (x,y) \in X_1 \times X_2$, $R(x,y)$ is trivial;
\item Every connected component of non-trivial constraints with scopes subsets of $X_1$ is a star.
\end{itemize}
\end{lemma}

\begin{proof}
Let $X_1 = \spxv \cup \{ \, z \in \mvar \, \mid \, \exists y \in
\spxv \, : \, \dxv{y} \times \dxv{z} \not\subseteq R(y,z) \, \}$.
We have $\spxv \subseteq X_1$, and by construction every non-trivial
constraint between $y \in X_1$ and $z \notin X_1$ must be such that
$y \notin \spxv$ and $y$ is adjacent to a variable in $\spxv$ via a
non-trivial constraint. By Lemma~\ref{lem:path} this is impossible,
and hence there is no non-trivial constraint between $X_1$ and $\mvar \backslash X_1$. The last property is immediate
by Lemma~\ref{lem:path}.
\end{proof}

\begin{theorem}
\label{thm:R8sac}
\CSP{\PatternREight} is solved by singleton arc consistency.
\end{theorem}

\begin{proof}
Let $I \in$ \CSP{\PatternREight} and suppose that $I$ is singleton
arc consistent. Let $x \in \mvar$ and $v \in D(x)$. Because of
singleton arc consistency the instance $\ixv$ has no empty domains. We remove all neighbourhood substitutable values from $\ixv$. By Lemma~\ref{lem:sep}, the variable set of $\ixv$ can be divided
into two parts $X_1,X_2$ such that $\ixv$ has a solution if and only
if both $\ixv[X_1]$ and $\ixv[X_2]$ are satisfiable. $\ixv[X_1]$ is
an arc consistent instance with no cycle of non-trivial constraints,
and hence is satisfiable. $\ixv[X_2]$ is exactly $I[X_2]$ with some neighbourhood substitutable values removed because no
variable in $X_2$ was affected by propagation after $x$ was assigned. Call this new instance $I[X_2]'$. Because $I[X_2]'$ is singleton arc consistent as well (being singleton arc consistent is invariant under projection and removal of neighbourhood substitutable values), we can repeat
the same reasoning on $I[X_2]'$. At each step the set $X_1$ cannot be
empty (it contains $x$) so this procedure will always terminate, and
because each $I[X_1]$ has a solution $I$ has a solution as well.
\end{proof}

Our proof of the SAC-solvability of \PatternRSevenMinus\ (Figure~\ref{fig:R7}) follows a similar reasoning, with two main differences. First, branching on just any variable-value pair (as we did for \PatternPThree\ and \PatternREight) may lead to a subproblem that is \textit{not} solved by arc consistency. However, once the right assignment is made the reward is much greater as all constraints involving a variable whose domain has been reduced by arc consistency must become trivial \textit{except at most one}.

Finding out which variable-value pair $(x,v)$ we should branch on is tricky. We first show that the above property is guaranteed to hold if $(x,v)$ is the meet point of the positive edges in a particular pattern \PatternMHat\ (Figure~\ref{fig:MV2}). However, \PatternMHat\ is an NP-hard pattern~\cite{cccms12:jair} so it might happen that \PatternMHat\ does not occur at all in the instance. To handle this problem we define a weaker pattern \PatternVTwo\ (Figure~\ref{fig:MV2}), whose absence is known to imply SAC-solvability (because it is a sub-pattern of \PatternREight) and hence can be safely assumed to occur somewhere. Our strategy is to branch on the assignment that corresponds to the meet point of the positive edges in \PatternVTwo\ and attempt to prove that the above property holds by induction, following the trace of the AC algorithm. We then show that if the induction started from
\PatternVTwo\ breaks then \PatternMHat\ must occur somewhere - a win-win situation.

\thicklines \setlength{\unitlength}{1.7pt}
\begin{figure}
\centering

\begin{picture}(200,50)
\put(0,0){
\begin{picture}(80,50)(0,0)
\put(10,30){\makebox(0,0){$\bullet$}}
\put(10,40){\makebox(0,0){$\bullet$}}

\put(40,10){\makebox(0,0){$\bullet$}}
\put(40,20){\makebox(0,0){$\bullet$}}
\put(40,30){\makebox(0,0){$\bullet$}}

\put(70,30){\makebox(0,0){$\bullet$}}
\put(70,40){\makebox(0,0){$\bullet$}}

\put(10,35){\oval(18,28)}
\put(40,20){\oval(18,38)}
\put(70,35){\oval(18,28)}

\dashline{4}(10,40)(40,30)
\dashline{4}(70,40)(40,30)
\dashline{4}(10,30)(40,20)
\dashline{4}(70,30)(40,10)

\put(10,30){\line(1,0){30}}
\put(40,30){\line(1,0){30}}
\end{picture}
}

\put(100,0){
\begin{picture}(80,50)(0,0)
\put(10,30){\makebox(0,0){$\bullet$}}

\put(40,10){\makebox(0,0){$\bullet$}}
\put(40,20){\makebox(0,0){$\bullet$}}
\put(40,30){\makebox(0,0){$\bullet$}}

\put(70,30){\makebox(0,0){$\bullet$}}

\put(10,30){\oval(18,18)}
\put(40,20){\oval(18,38)}
\put(70,30){\oval(18,18)}

\dashline{4}(10,30)(40,20)
\dashline{4}(70,30)(40,10)

\put(10,30){\line(1,0){30}}
\put(40,30){\line(1,0){30}}
\end{picture}
}
\end{picture}

\caption{The patterns \PatternMHat\ (left) and \PatternVTwo\ (right).} \label{fig:MV2}
\end{figure}
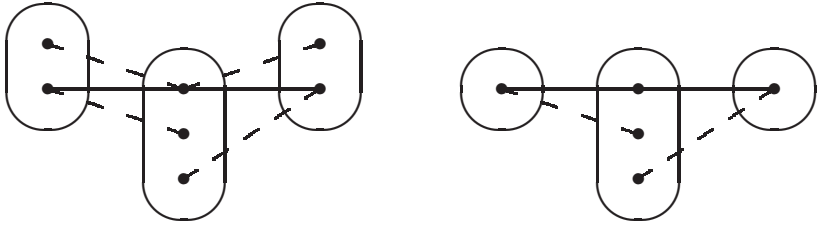

\begin{lemma}
\label{lem:moth}
Let $I=(\mvar,\mdom,\mcons) \in$ \CSP{\PatternRSevenMinus} be singleton arc consistent. Let $x \in \mvar$ be such that \PatternMHat\ occurs
on $(y,x,z)$ with $x$ the middle variable and $v$ be the value in
$D(x)$ that is the meet point of the two positive edges. Then every
constraint whose scope contains a variable in $\spxv$ is trivial in
$\ixv$, except possibly $R(y,z)$.
\end{lemma}

\begin{proof}
We prove the claim by induction. Every constraint with $p_0 = x$ in
its scope is trivial. Let $k \geq 0$ and suppose that the claim
holds for every constraint whose scope contains a variable in $\{p_i
\mid i \leq k\}$. Let $w \in \mvar \backslash \{p_i \mid i \leq k\}$
be a variable such that $R(p_{k+1},w)$ is not trivial in $\ixv$ and
$\{p_{k+1},w\} \neq \{y,z\}$. Let $p_j$ be such that $(p_j
\rightarrow p_{k+1}) \in \pxv$, with $j \leq k$. Because $R(p_{k+1},w)$ is not
trivial and $\ixv$ is arc consistent, there exist $a_{k+1},b_{k+1}
\in \dxv{p_{k+1}}$ and $a_w \in \dxv{w}$ such that $(a_{k+1},a_w) \in
R(p_{k+1},w)$ and $(b_{k+1},a_w) \notin R(p_{k+1},w)$.

If $R(p_j,p_{k+1})$ is trivial, then there exists $a_j \in
\dxv{p_j}$ such that $(a_j,a_{k+1}),(a_j,b_{k+1}) \in R(p_j,p_{k+1})$
and $c_{k+1}$ such that $(a_j,c_{k+1}) \notin R(p_j,p_{k+1})$
($c_{k+1}$ is one of the values that were eliminated by arc
consistency at step $(p_j \rightarrow p_{k+1})$). Then, if $p_j \neq
x$ there exists $p_i$, $i < j$ such that $(p_i \rightarrow p_j) \in
\pxv$. By arc consistency and because some propagation must have
taken place in the domain of $p_j$ at step $(p_i \rightarrow p_j)$,
there exists $a_i \in \dxv{p_i}$ and $b_j$ such that $(a_i,a_j) \in
R(p_i,p_j)$ and $(a_i,b_j) \notin R(p_i,p_j)$. It follows that
\PatternRSevenMinus\ occurs on $(p_i,p_j,p_{k+1},w)$, a
contradiction. On the other hand, if $p_j = x$ then we obtain the
same contradiction by using either $y$ or $z$ (the one which does not
appear in $\{p_{k+1},w\}$) instead of $p_i$.

By the induction hypothesis, if $R(p_j,p_{k+1})$ is not trivial then
$\{p_j,p_{k+1}\} = \{y,z\}$. By symmetry we can assume $p_{k+1} = z$.
$R(x,z)$ is trivial, so $\{ (v,a_{k+1}), (v,b_{k+1}) \} \subseteq
R(x,z)$. Furthermore, \PatternMHat\ occurs on $(y,x,z)$ so there
exist $c_{k+1}$ such that $(v,c_{k+1}) \notin R(x,z)$ and $a_y,b_x$
such that $(a_y,v) \in R(y,z)$ but $(a_y,b_x) \notin R(y,z)$. Then,
\PatternRSevenMinus\ occurs on $(y,x,z,w)$, a contradiction.

In both cases the induction holds, so the claim follows.
\end{proof}

\begin{lemma}
\label{lem:weakmoth} 
Let $I=(\mvar,\mdom,\mcons) \in$ \CSP{\PatternMHat} $\cap$ \CSP{\PatternRSevenMinus} be singleton arc consistent. Let $x \in \mvar$ be such that \PatternVTwo\ occurs on $(y,x,z)$ with $x$ the middle variable and $v$ be the value in $D(x)$ that is the meet point of the two positive edges.
Then every constraint whose scope contains a variable in $\spxv$ is trivial in $\ixv$,
except possibly $R(y,z)$. 
\end{lemma}

\begin{proof}
The proof follows the same idea as for Lemma~\ref{lem:moth}. However,
in this case the fact that \PatternMHat\ does not occur is critical
in order to keep the induction going.

Again, every constraint with $p_0 = x$ in its scope is trivial. Let
$k \geq 0$ and suppose that the claim holds for every constraint
whose scope contains a variable in $\{p_i \mid i \leq k\}$. Let $w
\in \mvar \backslash \{p_i \mid i \leq k\}$ be a variable such that
$R(p_{k+1},w) \in \mcons$ is not trivial in $\ixv$ and
$\{p_{k+1},w\} \neq \{y,z\}$. Let $p_j$ be such that $(p_j
\rightarrow p_{k+1}) \in \pxv$. Because $R(p_{k+1},w)$ is not
trivial and $\ixv$ is arc consistent, there exist $a_{k+1},b_{k+1}
\in \dxv{p_{k+1}}$ and $a_w \in \dxv{w}$ such that $(a_{k+1},a_w) \in
R(p_{k+1},w)$ and $(b_{k+1},a_w) \notin R(p_{k+1},w)$.

If $R(p_j,p_{k+1})$ is trivial we can proceed exactly as in the
proof of Lemma~\ref{lem:moth}, so let us focus on the case where
$R(p_j,p_{k+1})$ is not trivial. By induction we must have
$\{p_j,p_{k+1}\} = \{y,z\}$. We assume without loss of generality
that $p_{k+1} = z$. If $(x \rightarrow z) \in \pxv$ then we can use
$x$ instead of $y$ to bring us to the case where $R(p_j,p_{k+1})$ is
trivial, so let us assume $(x \rightarrow z) \notin \pxv$. Then, if
$(x \rightarrow y) \notin \pxv$ there exists $p_i,p_l$ such that $i,l
\leq k$, $(p_i \rightarrow y) \in \pxv$ and $(p_l \rightarrow p_i)
\in \pxv$. However, by induction $R(p_i,y)$ is trivial and thus
$R(y,z)$ should have been trivial as well (otherwise, the argument of Lemma~\ref{lem:moth} produces \PatternRSevenMinus\ on $(p_l,p_i,y,z)$). We can therefore assume
that $(x \rightarrow y) \in \pxv$ to work our way towards a
contradiction. In particular, this means that there exists $c_y$ such
that $(v,c_y) \notin R(x,y)$ ($c_y$ being a value eliminated by arc
consistency). Because \PatternVTwo\ occurs on $(y,x,z)$, there exists
$a_y \in \dxv{y}$ and $a_x$ such that $(v,a_y) \in R(x,y)$ and
$(a_x,a_y) \notin R(x,y)$. The picture below summarises the structure derived from the arguments above. Observe that we can always assume that
either $(a_y,a_{k+1}) \in R(y,z)$ or $(a_y,b_{k+1}) \in R(y,z)$ by
replacing $a_{k+1}$ or $b_{k+1}$ with a support for $a_y$ in
$R(y,z)$.

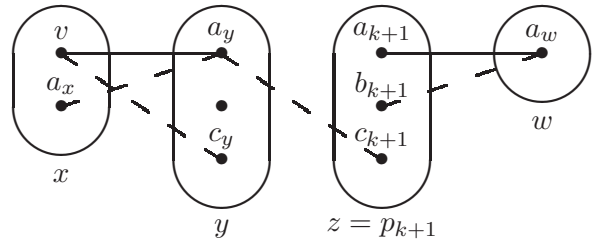
\thicklines \setlength{\unitlength}{2pt}
\begin{center}
\begin{picture}(110,40)(0,0)

\put(10,20){\makebox(0,0){$\bullet$}}
\put(10,24){\makebox(0,0){$a_x$}}
\put(10,30){\makebox(0,0){$\bullet$}}
\put(10,34){\makebox(0,0){$v$}}

\put(40,10){\makebox(0,0){$\bullet$}}
\put(40,14){\makebox(0,0){$c_y$}}
\put(40,20){\makebox(0,0){$\bullet$}}
\put(40,30){\makebox(0,0){$\bullet$}}
\put(40,34){\makebox(0,0){$a_y$}}

\put(70,10){\makebox(0,0){$\bullet$}}
\put(70,14){\makebox(0,0){$c_{k+1}$}}
\put(70,20){\makebox(0,0){$\bullet$}}
\put(70,24){\makebox(0,0){$b_{k+1}$}}
\put(70,30){\makebox(0,0){$\bullet$}}
\put(70,34){\makebox(0,0){$a_{k+1}$}}

\put(100,30){\makebox(0,0){$\bullet$}}
\put(100,34){\makebox(0,0){$a_w$}}

\put(10,25){\oval(18,28)}
\put(10,7){\makebox(0,0){$x$}}
\put(40,20){\oval(18,38)}
\put(40,-3){\makebox(0,0){$y$}}
\put(70,20){\oval(18,38)}
\put(70,-3){\makebox(0,0){$z = p_{k+1}$}}
\put(100,30){\oval(18,18)}
\put(100,17){\makebox(0,0){$w$}}

\put(10,30){\line(1,0){30}}
\put(70,30){\line(1,0){30}}

\dashline{4}(10,30)(40,10)
\dashline{4}(40,30)(10,20)
\dashline{4}(40,30)(70,10)
\dashline{4}(70,20)(100,30)

\end{picture}
\end{center}

\vspace{3mm}

If $(a_y,a_{k+1}) \in R(y,z)$, then unless \PatternRSevenMinus\
occurs on $(x,y,z,w)$ we must have $(a_y,b_{k+1}) \notin R(y,z)$. By
arc consistency of $\ixv$, there exists $b_y \in \dxv{y}$ such that
$(b_y,b_{k+1}) \in R(y,z)$, $(b_y,c_{k+1}) \notin R(y,z)$ (since
$c_{k+1}$ was eliminated by arc consistency) and because $R(x,y)$ is
trivial we have $(v,b_y) \in R(x,y)$. Again, unless
\PatternRSevenMinus\ occurs on $(x,y,z,w)$ we have $(b_y,a_{k+1})
\notin R(y,z)$. At this point one can observe in the picture below that the pattern
\PatternMHat\ occurs on $(x,y,z)$ with the meet point of the two
solid lines being $a_y$. This contradicts the assumption that $I \in$
\CSP{\PatternMHat}.

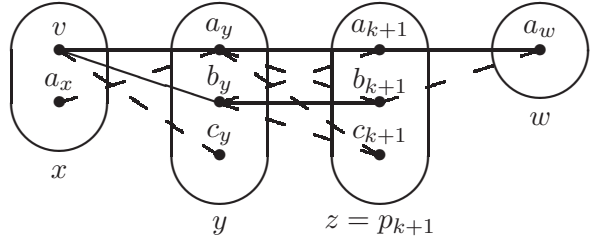
\vspace{5mm}
\begin{center}
\begin{picture}(110,40)(0,0)

\put(10,20){\makebox(0,0){$\bullet$}}
\put(10,24){\makebox(0,0){$a_x$}}
\put(10,30){\makebox(0,0){$\bullet$}}
\put(10,34){\makebox(0,0){$v$}}

\put(40,10){\makebox(0,0){$\bullet$}}
\put(40,14){\makebox(0,0){$c_y$}}
\put(40,20){\makebox(0,0){$\bullet$}}
\put(40,24.5){\makebox(0,0){$b_y$}}
\put(40,30){\makebox(0,0){$\bullet$}}
\put(40,34){\makebox(0,0){$a_y$}}

\put(70,10){\makebox(0,0){$\bullet$}}
\put(70,14){\makebox(0,0){$c_{k+1}$}}
\put(70,20){\makebox(0,0){$\bullet$}}
\put(70,24){\makebox(0,0){$b_{k+1}$}}
\put(70,30){\makebox(0,0){$\bullet$}}
\put(70,34){\makebox(0,0){$a_{k+1}$}}

\put(100,30){\makebox(0,0){$\bullet$}}
\put(100,34){\makebox(0,0){$a_w$}}

\put(10,25){\oval(18,28)}
\put(10,7){\makebox(0,0){$x$}}
\put(40,20){\oval(18,38)}
\put(40,-3){\makebox(0,0){$y$}}
\put(70,20){\oval(18,38)}
\put(70,-3){\makebox(0,0){$z = p_{k+1}$}}
\put(100,30){\oval(18,18)}
\put(100,17){\makebox(0,0){$w$}}

\put(10,30){\line(1,0){30}}
\put(70,30){\line(1,0){30}}
\put(40,30){\line(1,0){30}}
\put(40,20){\line(1,0){30}}
\put(10,30){\line(3,-1){30}}

\dashline{4}(10,30)(40,10)
\dashline{4}(40,30)(10,20)
\dashline{4}(40,20)(70,10)
\dashline{4}(40,20)(70,30)
\dashline{4}(40,30)(70,20)
\dashline{4}(40,30)(70,10)
\dashline{4}(70,20)(100,30)

\end{picture}
\end{center}

\vspace{3mm}

The case where $(a_y,b_{k+1}) \in R(y,z)$ is almost symmetric.
Because \PatternRSevenMinus\ does not occur, we must have
$(a_y,a_{k+1}) \notin R(y,z)$. By arc consistency, there exists some
$b_y \in \dxv{y}$ such that $(b_y,a_{k+1}) \in R(y,z)$,
$(b_y,c_{k+1}) \notin R(y,z)$ and because $R(x,y)$ is trivial we
have $(v,b_y) \in R(x,y)$. It follows from the absence of
\PatternRSevenMinus\ that $(b_y,b_{k+1}) \notin R(y,z)$, which create
the pattern \PatternMHat\ on $(x,y,z)$ with its meet point being
$a_y$, as shown in the picture below.

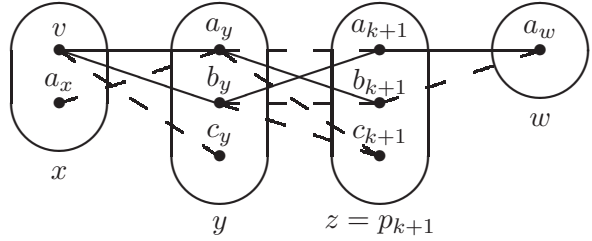
\vspace{5mm}
\begin{center}
\begin{picture}(110,40)(0,0)

\put(10,20){\makebox(0,0){$\bullet$}}
\put(10,24){\makebox(0,0){$a_x$}}
\put(10,30){\makebox(0,0){$\bullet$}}
\put(10,34){\makebox(0,0){$v$}}

\put(40,10){\makebox(0,0){$\bullet$}}
\put(40,14){\makebox(0,0){$c_y$}}
\put(40,20){\makebox(0,0){$\bullet$}}
\put(40,24.5){\makebox(0,0){$b_y$}}
\put(40,30){\makebox(0,0){$\bullet$}}
\put(40,34){\makebox(0,0){$a_y$}}

\put(70,10){\makebox(0,0){$\bullet$}}
\put(70,14){\makebox(0,0){$c_{k+1}$}}
\put(70,20){\makebox(0,0){$\bullet$}}
\put(70,24){\makebox(0,0){$b_{k+1}$}}
\put(70,30){\makebox(0,0){$\bullet$}}
\put(70,34){\makebox(0,0){$a_{k+1}$}}

\put(100,30){\makebox(0,0){$\bullet$}}
\put(100,34){\makebox(0,0){$a_w$}}

\put(10,25){\oval(18,28)}
\put(10,7){\makebox(0,0){$x$}}
\put(40,20){\oval(18,38)}
\put(40,-3){\makebox(0,0){$y$}}
\put(70,20){\oval(18,38)}
\put(70,-3){\makebox(0,0){$z = p_{k+1}$}}
\put(100,30){\oval(18,18)}
\put(100,17){\makebox(0,0){$w$}}

\put(10,30){\line(1,0){30}}
\put(70,30){\line(1,0){30}}
\put(40,30){\line(3,-1){30}}
\put(40,20){\line(3,1){30}}
\put(10,30){\line(3,-1){30}}

\dashline{4}(10,30)(40,10)
\dashline{4}(40,30)(10,20)
\dashline{4}(40,20)(70,10)
\dashline{4}(40,20)(70,20)
\dashline{4}(40,30)(70,30)
\dashline{4}(40,30)(70,10)
\dashline{4}(70,20)(100,30)

\end{picture}
\end{center}

\vspace{3mm}

This final contradiction completes the proof.
\end{proof}

\begin{theorem}
\label{thm:R7sac}
\CSP{\PatternRSevenMinus} is solved by singleton arc consistency.
\end{theorem}

\begin{proof}
Let $I = (\mvar,\mdom,\mcons) \in$ \CSP{\PatternRSevenMinus}, and
suppose that enforcing SAC did not lead to a wipeout of any variable
domain. If \PatternVTwo\ does not occur in $I$ then it has a solution
(recall that absence of \PatternVTwo\ ensures solvability by SAC), so
let us assume that \PatternVTwo\ occurs. Let $x \in \mvar$ and $v \in
D(x)$ be such that $v$ is the meet point of solid edges of
\PatternMHat\ if \PatternMHat\ occurs in $I$, and the meet point of
\PatternVTwo\ otherwise. $I$ is SAC so the instance $\ixv$ has no
empty domains. By Lemma~\ref{lem:moth} and Lemma~\ref{lem:weakmoth},
there is at most one non-trivial constraint in $\ixv[\spxv]$ so by
arc consistency for every $x_1 \in \spxv$ and $v_1 \in \dxv{x_1}$ there
is a solution $\phi$ to $\ixv[\spxv]$ such that $\phi(x_1) = v_1$.
Furthermore, $\ixv[\mvar \backslash \spxv] = I[\mvar \backslash
\spxv]$ and there is at most one non-trivial constraint in $\ixv$ with one
endpoint in $\spxv$ and the other in $\mvar \backslash \spxv$. By
combining the two properties we obtain that $\ixv$ has a solution if
and only if $I[\mvar \backslash \spxv]$ has one. Because $I[\mvar
\backslash \spxv]$ is SAC and \PatternRSevenMinus\ still does not
occur, we can repeat the operation until we have a solution to the
whole instance.
\end{proof}

\section{Tractability of \PatternQThree\ and \PatternRFive}
\label{sec:Q3R5}

For our last two proofs of SAC-decidability, we depart from the trace technique. Our fundamental goal, however, remains the same: find an operation which shrinks the instance without altering satisfiability, introducing the pattern or losing singleton arc consistency. For \PatternQThree\ this operation is BTP-merging~\cite{cooper16:ai} and for \PatternRFive\ it is removing constraints.

Consider the pattern V$^{-}$ shown in Figure~\ref{fig:VMinus}(a).
We say that V$^{-}$ occurs at point $a$ or at variable $x$ if $a \in D(x)$
is the central point of the pattern in the instance. The pattern
V$^{-}$ is known to be tractable since all instances in
\CSP{V$^{-}$}
satisfy the joint-winner property~\cite{cz11:ai}. However, we show a slightly different result,
namely that singleton arc consistency is sufficient to solve
instances in which V$^{-}$ only occurs at degree-2 variables.

\thicklines \setlength{\unitlength}{1.3pt}
\begin{figure}
\centering

\begin{picture}(200,100)(0,0) 
\put(0,0){
\begin{picture}(110,70)(0,0)
\put(0,50){\makebox(0,0){$\bullet$}} 
\put(0,50){\oval(18,22)}
\put(40,70){\makebox(0,0){$\bullet$}} 
\put(40,70){\oval(18,22)}
\put(40,30){\makebox(0,0){$\bullet$}} 
\put(40,30){\oval(18,22)}
\dashline{3}(0,50)(40,70) 
\dashline{3}(0,50)(40,30)
\put(15,5){\makebox(0,0){(a)}}
\end{picture}
}

\ignore{
        \put(80,0){
        \begin{picture}(160,100)(0,0)
        \put(20,90){\makebox(0,0){$\bullet$}}
        \put(20,70){\makebox(0,0){$\bullet$}} \put(20,80){\oval(18,38)}
        \put(60,10){\makebox(0,0){$\bullet$}} \put(60,20){\oval(18,38)}
        \put(100,90){\makebox(0,0){$\bullet$}}
        \put(100,70){\makebox(0,0){$\bullet$}} \put(100,80){\oval(18,38)}        
        \put(140,10){\makebox(0,0){$\bullet$}} \put(140,20){\oval(18,38)}        
        \dashline{3}(20,70)(100,70) \dashline{3}(20,90)(60,10)
        \dashline{3}(100,90)(140,10) \put(20,70){\line(4,1){80}}
        \put(20,90){\line(4,-1){80}} \put(20,70){\line(2,-3){40}}
        \put(100,70){\line(2,-3){40}}
        \put(9,60){\makebox(0,0){$x$}} \put(60,10){\line(2,3){40}}
        \put(60,10){\line(1,0){80}} \put(20,70){\line(2,-1){120}}
        \end{picture}
        }
}

\put(100,0){\begin{picture}(100,90)(10,15)
\put(20,70){\makebox(0,0){$\bullet$}} \put(20,70){\oval(18,22)}
\put(60,30){\makebox(0,0){$\bullet$}}
\put(60,50){\makebox(0,0){$\bullet$}} \put(60,40){\oval(18,38)}
\put(100,70){\makebox(0,0){$\bullet$}} \put(100,70){\oval(18,22)}
\dashline{3}(20,70)(60,50) \dashline{3}(60,30)(100,70)
\put(20,70){\line(1,-1){40}} \put(60,50){\line(2,1){40}}
\put(15,70){\makebox(0,0){$c$}} \put(60,25){\makebox(0,0){$b$}}
\put(60,45){\makebox(0,0){$a$}} \put(105,70){\makebox(0,0){$d$}}
\put(100,55){\makebox(0,0){$z$}} \put(60,63.5){\makebox(0,0){$x$}}
\put(20,54){\makebox(0,0){$y$}} \put(20,70){\line(1,0){80}}
\put(15,20){\makebox(0,0){(b)}}
\end{picture}}

\end{picture}

\caption{(a) The pattern V$^{-}$ and (b) the associated broken-triangle pattern (BTP).}
\label{fig:VMinus}
\label{fig:BTP}
\end{figure}
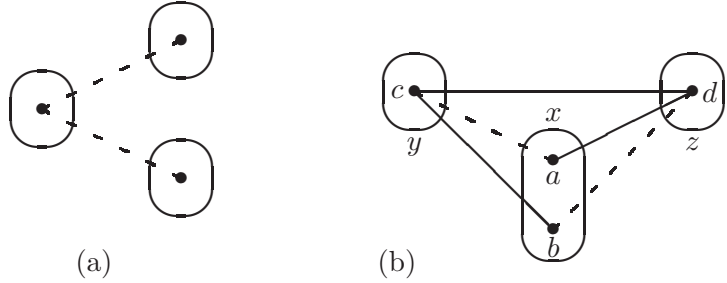

\begin{lemma} 
\label{lem:V-}
Instances in which V$^{-}$ only occurs at degree-2 variables are
solved by singleton arc consistency.
\end{lemma}

\begin{proof}
Singleton arc consistency only eliminates values from domains and
thus cannot increase the degree of a variable nor introduce the
pattern V$^{-}$. Hence, singleton arc consistency cannot lead to the
occurrence of the pattern V$^{-}$ at a variable of degree greater
than two. Therefore it is sufficient to show that any SAC instance
$I$ in which V$^{-}$ only occurs at degree-2 variables is
satisfiable.

We will show that it is always possible to find an independent
partial solution, i.e. an assignment to a non-empty subset of the
variables of $I$ which is compatible with all possible assignments to
the other variables. A solution can be found by repeatedly finding
independent partial solutions. If $I$ has only degree-2 variables,
then it is folklore (and easy to show) that singleton arc consistency implies
satisfiability.
So we only need to consider the case
in which $I$ has at least one variable $x_1$ of degree greater than
or equal to three. Choose an arbitrary value $a_1 \in D(x_1)$. If
this assignment is compatible with all assignments to all other
variables, then this is the required independent partial solution, so
suppose that there is a negative edge $a_1b$ where $b \in D(x_2)$ for
some variable $x_2$. By assumption, since $x_1$ has degree greater
than or equal to three, the pattern V$^{-}$ does not occur at $x_1$
and hence the assignment $(x_1,a_1)$ is compatible with all
assignments to all variables other than $x_1,x_2$.

\thicklines \setlength{\unitlength}{1.5pt}
\begin{figure}
\centering
\begin{picture}(220,70)(0,0)
\put(10,30){\makebox(0,0){$\bullet$}} \put(10,40){\oval(18,38)}
\put(60,30){\makebox(0,0){$\bullet$}}
\put(60,50){\makebox(0,0){$\bullet$}} \put(60,40){\oval(18,38)}
\put(110,30){\makebox(0,0){$\bullet$}}
\put(110,50){\makebox(0,0){$\bullet$}} \put(110,40){\oval(18,38)}
\put(160,30){\makebox(0,0){$\bullet$}}
\put(160,50){\makebox(0,0){$\bullet$}} \put(160,40){\oval(18,38)}
\put(210,30){\makebox(0,0){$\bullet$}}
\put(210,50){\makebox(0,0){$\bullet$}} \put(210,40){\oval(18,38)}
\dashline{3}(10,30)(60,50) \dashline{3}(60,30)(110,50)
\dashline{3}(110,30)(160,50) \dashline{3}(160,30)(210,50)
\put(10,30){\line(1,0){200}} \put(57,26){\makebox(0,0){$a_i$}}
\put(157,25.5){\makebox(0,0){$a_j$}} \put(60,16){\makebox(0,0){$x_i$}}
\put(160,16){\makebox(0,0){$x_j$}}
\end{picture}
\caption{The edge $a_ia_j$ must be positive, otherwise the pattern
V$^{-}$ would occur at $a_i$ and variable $x_i$ would have degree at
least three. In the special case $i=1$, this follows from our choice
of $x_1$ to be a variable of degree at least three.} \label{fig:pfV-}
\end{figure}
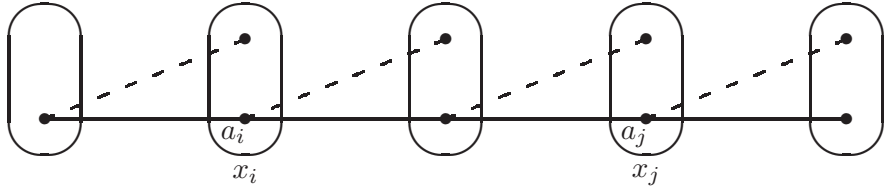

Now suppose that we have a partial assignment
$(x_1,a_1),\ldots,(x_k,a_k)$, as shown in Figure~\ref{fig:pfV-}, such
that
\begin{enumerate}
\item for $i=1,\ldots,k$, $a_i \in D(x_i)$,
\item for $i=1,\ldots,k-1$, $\exists b \in D(x_{i+1})$ such that
  $a_ib$ is a negative edge.
\item for $i=1,\ldots,k-1$, $a_ia_{i+1}$ is a positive edge,
\end{enumerate}
The assignments $(x_i,a_i)$ ($i=1,\ldots,k$) are all compatible with
each other, otherwise the pattern V$^{-}$ would occur at a variable
of degree three or greater, as illustrated in Figure~\ref{fig:pfV-}.
Furthermore, for the same reason, the assignments $(x_i,a_i)$
($i=1,\ldots,k-1$) are all compatible with all possible assignments
to all variables other than $x_1,\ldots,x_k$. It only remains to
consider the compatibility of $a_k$ with the assignments to variables
other than $x_1,\ldots,x_k$.

If the assignment $(x_k,a_k)$ is compatible with all assignments to
all variables other than $x_1,\ldots,x_k$, then we have an
independent partial solution to $x_1,\ldots,x_k$. On the other hand,
if for some $x_{k+1} \notin \{x_1,\ldots,x_k\}$, there is $b \in
D(x_{k+1})$ such that $a_kb$ is a negative edge, then by (singleton)
arc consistency there exists $a_{k+1} \in D(x_{k+1})$ such that
$a_ka_{k+1}$ is a positive edge and we have a larger partial
assignment with the above three properties. Therefore, we can always
add another assignment until the resulting partial assignment is an
independent partial solution (or we have assigned all variables).
\end{proof}

Two values $a,b \in D(x)$ are \emph{BTP-mergeable}~\cite{cooper16:ai} if
there are not two other distinct variables $y,z \neq x$ such that there exist
$c \in D(y)$ and $d \in D(z)$ with $ad,bc,cd$ positive
edges and $ac,bd$ negative edges as shown in Figure~\ref{fig:BTP}(b).
The \emph{BTP-merging} operation consists in merging two BTP-mergeable points $a,b
\in D(x)$: the points $a,b$ are replaced by a
new point $c$ in $D(x)$ such that for all other variables $w \neq x$
and for all $d \in D(w)$, $cd$ is a positive edge if at least one of
$ad,bd$ was a positive edge (a negative edge otherwise). BTP-merging preserves satisfiability~\cite{cooper16:ai}.

\begin{lemma} 
\label{lem:BTP}
Let $P$ be a pattern in which no point occurs in more than one
positive edge. Then the BTP-merging operation cannot introduce the
pattern $P$ in an instance $I \in$ \CSP{P}.
\end{lemma}

\begin{proof}
Suppose that the pattern $P$ occurs in an instance $I'$ obtained by
BTP-merging of two points $a,b$ in $I$ to create a new point $c$ in
$I'$. From the assumptions about $P$, we know that $c$ belongs to any
number of negative edges $ce_1,\ldots,ce_r$, but at most one positive
edge $cd$ in the occurrence of $P$ in $I'$. By the definition of
merging, in $I$ one of $ad,bd$ must have been a positive edge and all
of $ae_1,\ldots,ae_r$ and $be_1,\ldots,be_r$ must have been negative.
Without loss of generality, suppose that $ad$ was a positive edge.
But then the pattern $P$ occurred in $I$ (on $a$ instead of $c$)
which is a contradiction.
\end{proof}

Since \PatternQThree\ has no point which occurs in more than one positive edge, we
can deduce from Lemma~\ref{lem:BTP} that \PatternQThree\ cannot be introduced by
BTP-merging. We then combine this property with Lemma~\ref{lem:V-} by proving that V$^-$ can only occur at degree-2 variables in any instance of \CSP{\PatternQThree} with no BTP-mergeable values.

\begin{theorem}
\label{thm:Q3sac}
\CSP{\PatternQThree} is solved by singleton arc consistency.
\end{theorem}

\begin{proof}
Let $I \in$ \CSP{\PatternQThree}. Since establishing singleton arc
consistency cannot introduce patterns, and hence in particular cannot
introduce \PatternQThree, we can assume that $I$ is SAC. Let $I'$ be the result
of applying BTP-merging operations to $I$ until convergence. By
Lemma~\ref{lem:BTP}, we know that $I' \in$ \CSP{\PatternQThree}.
Furthermore, since BTP-merging only weakens constraints (in the sense
that the new value $c$ is constrained less than either of the values
$a,b$ it replaces), it cannot destroy singleton arc consistency;
hence $I'$ is SAC. By Lemma~\ref{lem:V-}, it suffices to show that
V$^{-}$ cannot occur in $I'$ at variables of degree three or greater.

\thicklines \setlength{\unitlength}{1.5pt}
\begin{figure}
\centering
\begin{picture}(220,90)(0,0)

\put(60,0){\begin{picture}(100,90)(10,0)
\put(20,70){\makebox(0,0){$\bullet$}} \put(20,70){\oval(18,22)}
\put(60,30){\makebox(0,0){$\bullet$}}
\put(60,10){\makebox(0,0){$\bullet$}}
\put(60,50){\makebox(0,0){$\bullet$}} \put(60,30){\oval(18,58)}
\put(100,70){\makebox(0,0){$\bullet$}} \put(100,70){\oval(18,22)}
\put(100,10){\makebox(0,0){$\bullet$}} \put(100,10){\oval(18,22)}
\dashline{3}(20,70)(60,50) \dashline{3}(60,30)(100,70)
\dashline{3}(60,10)(100,10) \put(20,70){\line(1,-1){40}}
\put(60,50){\line(2,1){40}} \put(15,70){\makebox(0,0){$c$}}
\put(60,25.5){\makebox(0,0){$b$}} \put(60,6){\makebox(0,0){$e$}}
\put(60,46){\makebox(0,0){$a$}} \put(105,70){\makebox(0,0){$d$}}
\put(105,10){\makebox(0,0){$f$}} \put(100,24){\makebox(0,0){$w$}}
\put(100,56){\makebox(0,0){$z$}} \put(60,63){\makebox(0,0){$x$}}
\put(20,55){\makebox(0,0){$y$}} \put(20,70){\line(1,0){80}}
\end{picture}}

\end{picture}
\caption{The pattern
V$^{-}$ cannot occur at $a$ if \PatternQThree\ does not occur in the instance.}
\label{fig:Q3pf}
\end{figure}
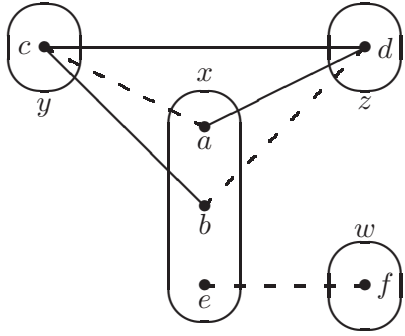

Consider an arbitrary point $a \in D(x)$ at a variable $x$ which is
of degree three or greater. We will show that V$^{-}$ cannot occur at
$a$, which will complete the proof. If $a$ belongs to no negative
edge then clearly V$^{-}$ cannot occur at $a$. The existence of a
negative edge and the (singleton) arc consistency of $I$ implies that
there there is some other value $b \in D(x)$. Since $a,b$ cannot be
BTP-merged, there must be other variables $y,z$ and values $c \in
D(y)$, $d \in D(z)$ with $ad,bc,cd$ positive edges and $ac,bd$
negative edges, as shown in Figure~\ref{fig:Q3pf}. Now since \PatternQThree\
does not occur in $I$, we can deduce that $a$ and $b$ are connected
by positive edges to all points in $D(v)$ for $v \notin \{x,y,z\}$.
Since $x$ is of degree three or greater, there must therefore be
another point $e \in D(x) \setminus \{a,b\}$ and a negative edge $ef$
where $f \in D(w)$ for some $w \notin \{x,y,z\}$ (as shown in
Figure~\ref{fig:Q3pf}). By applying the same argument as above,
knowing that $a,e$ cannot be BTP-merged, we can deduce that $a$ and
$e$ are connected by positive edges to all points in $D(v)$ for $v
\notin \{x,y,w\}$. Hence, $a$ can only be connected by negative edges
to points in $D(y)$. It follows that the pattern V$^{-}$ cannot occur
at $a$, which completes the proof.
\end{proof}

That only leaves \PatternRFive. Removing constraints cannot introduce \PatternRFive\ because it is a monotone pattern, so we can apply repeatedly the following lemma to obtain our last result.

\thicklines \setlength{\unitlength}{1.5pt}
\begin{figure}
\centering
\begin{picture}(160,230)(0,0)

\put(0,120){
\begin{picture}(210,100)(0,0)
\put(20,90){\makebox(0,0){$\bullet$}}
\put(20,70){\makebox(0,0){$\bullet$}} \put(20,80){\oval(18,38)}
\put(60,30){\makebox(0,0){$\bullet$}}
\put(60,10){\makebox(0,0){$\bullet$}} \put(60,20){\oval(18,38)}
\put(100,90){\makebox(0,0){$\bullet$}}
\put(100,70){\makebox(0,0){$\bullet$}} \put(100,80){\oval(18,38)}
\put(140,10){\makebox(0,0){$\bullet$}} \put(140,20){\oval(18,38)}
\dashline{3}(20,70)(100,70) \dashline{3}(20,90)(60,10)
\put(20,70){\line(4,1){80}}
\put(20,90){\line(4,-1){80}} \put(20,70){\line(2,-3){40}}
\put(100,70){\line(2,-3){40}}
\put(15,70){\makebox(0,0){$a$}} \put(15,90){\makebox(0,0){$c$}}
\put(105,70){\makebox(0,0){$b$}} \put(105,90){\makebox(0,0){$d$}}
\put(55,10){\makebox(0,0){$e$}} \put(145,10){\makebox(0,0){$f$}}
\put(55,30){\makebox(0,0){$g$}} \put(20,70){\line(1,-1){40}} \put(60,30){\line(2,3){40}}
\put(9,60){\makebox(0,0){$x$}} \put(111,60){\makebox(0,0){$y$}}
\put(49,0){\makebox(0,0){$w$}} \put(151,0){\makebox(0,0){$z$}}
\put(180,50){\makebox(0,0){\shortstack{$a = s[x]$ \\ $b = s[y]$ \\ $e = s[w]$ \\ $f = s[z]$}}}
\put(10,5){\makebox(0,0){(a)}}
\end{picture}}

\put(0,0){
\begin{picture}(210,100)(0,0)
\put(20,90){\makebox(0,0){$\bullet$}}
\put(20,70){\makebox(0,0){$\bullet$}} \put(20,80){\oval(18,38)}
\put(60,10){\makebox(0,0){$\bullet$}} \put(60,20){\oval(18,38)}
\put(100,90){\makebox(0,0){$\bullet$}}
\put(100,70){\makebox(0,0){$\bullet$}} \put(100,80){\oval(18,38)}
\put(140,10){\makebox(0,0){$\bullet$}} \put(140,20){\oval(18,38)}
\dashline{3}(20,70)(100,70) \dashline{3}(20,90)(60,10)
\put(20,70){\line(4,1){80}}
\put(20,90){\line(4,-1){80}} \put(20,70){\line(2,-3){40}}
\put(100,70){\line(2,-3){40}}
\put(15,70){\makebox(0,0){$d$}} \put(15,90){\makebox(0,0){$b$}}
\put(105,70){\makebox(0,0){$e$}} \put(105,90){\makebox(0,0){$g$}}
\put(55,10){\makebox(0,0){$a$}} \put(145,10){\makebox(0,0){$f$}}
\put(9,60){\makebox(0,0){$y$}} \put(111,60){\makebox(0,0){$w$}}
\put(49,0){\makebox(0,0){$x$}} \put(151,0){\makebox(0,0){$z$}}
\put(180,50){\makebox(0,0){\shortstack{$a = s[x]$ \\ $b = s[y]$ \\ $e = s[w]$ \\ $f = s[z]$}}}
\put(10,5){\makebox(0,0){(b)}}
\end{picture}}

\end{picture}
\caption{Since the pattern \PatternRFive\ does not occur in $I$, we can deduce that
(a) $df$ is a positive edge, and (b) $gf$ is a positive edge.} \label{fig:R5proof}
\end{figure}
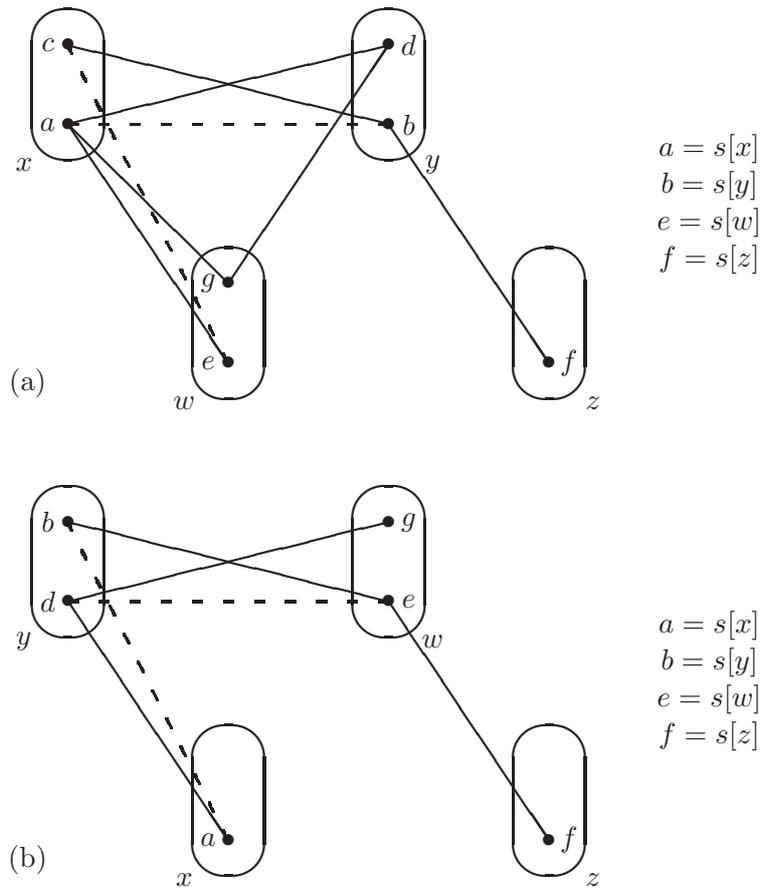

\begin{lemma} 
\label{lem:R5}
If the pattern \PatternRFive\ does not occur in a singleton arc consistent
binary CSP instance $I$, then removing any
constraint leaves the satisfiability of $I$ invariant.
\end{lemma}

\begin{proof}
Suppose that the pattern \PatternRFive\ does not occur in the instance $I$ and
that $I$ is singleton arc consistent. Let $I'$ be the instance which
results when we eliminate the constraint between an arbitrary pair of variables $x$ and
$y$. Suppose that $s$ is a solution to $I'$. It suffices to exhibit a
solution to $I$. We use $s[z]$ to denote the value assigned to
variable $z$ in $s$. Let $a=s[x]$ and $b=s[y]$. If $ab$ is a positive edge in $I$
than $s$ is also a solution to $I$, so we assume that $ab$ is a negative edge.

By (singleton) arc consistency, there exists $c \in D(x)$ such that $bc$ is
a positive edge. Either we can replace $a$ in $s$ by $c$ to produce a
solution to $I$, or there is some variable $w \notin \{x,y\}$ such
that $ce$ is a negative edge where $e = s[w]$. By singleton arc
consistency, there exist $d \in D(y)$ and $g \in D(w)$ such that
$ad$, $ag$ and $dg$ are positive edges. Consider any variable $z
\notin \{x,y,w\}$. We have the situation in $I$ shown in
Figure~\ref{fig:R5proof}(a) where $f = s[z]$. The positive edges $ae$ and $bf$
follow from the fact that $s$ is a solution to $I'$. Now, since \PatternRFive\ does not
occur in $I$, we can deduce that $df$ is a positive edge.
Recall that the variable $z$ was any variable other than $x,y$ or $w$.

Since $d \in D(y)$ is compatible with $a=s[x]$,
we have just shown that $d$ can only be incompatible with $s[z]$ when
$z=w$. Thus, either we can replace $b$ by $d$ to produce a solution
to $I$, or $de$ is a negative edge. In this latter case, consider any
variable $z \notin \{x,y,w\}$ and again denote $s[z]$ by $f$. We have the situation in $I$ shown in
Figure~\ref{fig:R5proof}(b). The positive edges $be$ and $ef$
follow from the fact that $s$ is a solution to $I'$.
Since the pattern \PatternRFive\ does not occur in
$I$, we can deduce that $gf$ is a positive edge.
But then we can replace $b$ by $d$ and $e$ by $g$ in $s$ to produce a
solution to $I$.
\end{proof}

Note that Lemma~\ref{lem:R5} is technically true for all SAC-solvable patterns (not only \PatternRFive); this is simply the only case where we are able to prove it directly. 

\begin{theorem}
\label{thm:R5sac}
\CSP{\PatternRFive} is solved by singleton arc consistency.
\end{theorem}

\begin{proof}
Establishing singleton arc consistency preserves satisfiability and cannot introduce any pattern
and, hence in particular, cannot introduce \PatternRFive.
Consider a SAC instance $I \in$ \CSP{\PatternRFive} which has
non-empty domains.
By Lemma~\ref{lem:R5}, we can eliminate any constraint.
The resulting instance is still SAC. Furthermore, \PatternRFive\ has not been introduced since \PatternRFive\
is monotone. Therefore, we can keep on eliminating constraints until all
constraints have been eliminated. The resulting instance is trivially
satisfiable and hence so was the original instance $I$.
It follows that singleton arc consistency decides all instances in \CSP{\PatternRFive}.
\end{proof}

\section{A Necessary Condition for Solvability by SAC}
\label{sec:CounterExamples}

\thicklines \setlength{\unitlength}{1.5pt}
\begin{figure}
\centering
\begin{picture}(250,195)(0,0)

\put(0,0){
\begin{picture}(150,195)(0,-5)

\put(20,20){\oval(18,38)} \put(20,70){\oval(18,38)}
\put(20,120){\oval(18,38)} \put(20,170){\oval(18,38)}
\put(120,35){\oval(18,48)} \put(120,95){\oval(18,48)}
\put(120,155){\oval(18,48)} \put(20,10){\makebox(0,0){$\bullet$}}
\put(20,20){\makebox(0,0){$\bullet$}}
\put(20,30){\makebox(0,0){$\bullet$}}
\put(20,60){\makebox(0,0){$\bullet$}}
\put(20,70){\makebox(0,0){$\bullet$}}
\put(20,80){\makebox(0,0){$\bullet$}}
\put(20,110){\makebox(0,0){$\bullet$}}
\put(20,120){\makebox(0,0){$\bullet$}}
\put(20,130){\makebox(0,0){$\bullet$}}
\put(20,160){\makebox(0,0){$\bullet$}}
\put(20,170){\makebox(0,0){$\bullet$}}
\put(20,180){\makebox(0,0){$\bullet$}}
\put(120,20){\makebox(0,0){$\bullet$}}
\put(120,30){\makebox(0,0){$\bullet$}}
\put(120,40){\makebox(0,0){$\bullet$}}
\put(120,50){\makebox(0,0){$\bullet$}}
\put(120,80){\makebox(0,0){$\bullet$}}
\put(120,90){\makebox(0,0){$\bullet$}}
\put(120,100){\makebox(0,0){$\bullet$}}
\put(120,110){\makebox(0,0){$\bullet$}}
\put(120,140){\makebox(0,0){$\bullet$}}
\put(120,150){\makebox(0,0){$\bullet$}}
\put(120,160){\makebox(0,0){$\bullet$}}
\put(120,170){\makebox(0,0){$\bullet$}}
\put(0,170){\makebox(0,0){$x_1$}} \put(0,120){\makebox(0,0){$x_2$}}
\put(0,70){\makebox(0,0){$x_3$}} \put(0,20){\makebox(0,0){$x_4$}}
\put(140,155){\makebox(0,0){$y_1$}}
\put(140,95){\makebox(0,0){$y_2$}} \put(140,35){\makebox(0,0){$y_3$}}

\dashline{5}(20,180)(120,160) \dashline{5}(20,180)(120,150)
\dashline{5}(20,180)(120,140) \dashline{5}(20,170)(120,100)
\dashline{5}(20,170)(120,90) \dashline{5}(20,170)(120,80)
\dashline{5}(20,160)(120,40) \dashline{5}(20,160)(120,30)
\dashline{5}(20,160)(120,20)

\dashline{5}(20,130)(120,170) \dashline{5}(20,130)(120,150)
\dashline{5}(20,130)(120,140) \dashline{5}(20,120)(120,110)
\dashline{5}(20,120)(120,90) \dashline{5}(20,120)(120,80)
\dashline{5}(20,110)(120,50) \dashline{5}(20,110)(120,30)
\dashline{5}(20,110)(120,20)

\dashline{5}(20,80)(120,170) \dashline{5}(20,80)(120,160)
\dashline{5}(20,80)(120,140) \dashline{5}(20,70)(120,110)
\dashline{5}(20,70)(120,100) \dashline{5}(20,70)(120,80)
\dashline{5}(20,60)(120,50) \dashline{5}(20,60)(120,40)
\dashline{5}(20,60)(120,20)

\dashline{5}(20,30)(120,170) \dashline{5}(20,30)(120,160)
\dashline{5}(20,30)(120,150) \dashline{5}(20,20)(120,110)
\dashline{5}(20,20)(120,100) \dashline{5}(20,20)(120,90)
\dashline{5}(20,10)(120,50) \dashline{5}(20,10)(120,40)
\dashline{5}(20,10)(120,30)

\put(70,0){\makebox(0,0){(a)}}
\end{picture}}

\put(180,80){
\begin{picture}(70,50)(0,0)
\put(10,25){\oval(18,48)} \put(60,25){\oval(18,48)}
\put(10,55){\makebox(0,0){$x_1$}} \put(60,55){\makebox(0,0){$x_2$}}
\put(10,10){\makebox(0,0){$\bullet$}} \dashline{4}(10,40)(60,30)
\put(10,20){\makebox(0,0){$\bullet$}} \dashline{4}(10,40)(60,20)
\put(10,30){\makebox(0,0){$\bullet$}} \dashline{4}(10,40)(60,10)
\put(10,40){\makebox(0,0){$\bullet$}} \dashline{4}(10,30)(60,40)
\put(60,10){\makebox(0,0){$\bullet$}} \dashline{4}(10,20)(60,40)
\put(60,20){\makebox(0,0){$\bullet$}} \dashline{4}(10,10)(60,40)
\put(60,30){\makebox(0,0){$\bullet$}}
\put(60,40){\makebox(0,0){$\bullet$}} \put(35,0){\makebox(0,0){(b)}}
\end{picture}}

\end{picture}
\caption{(a) The instance $I_{3,4}$ which is SAC but has no solution.
(b) The constraint between variables $x_1$ and $x_2$ in instance
$I_5$.} \label{fig:I34}
\end{figure}
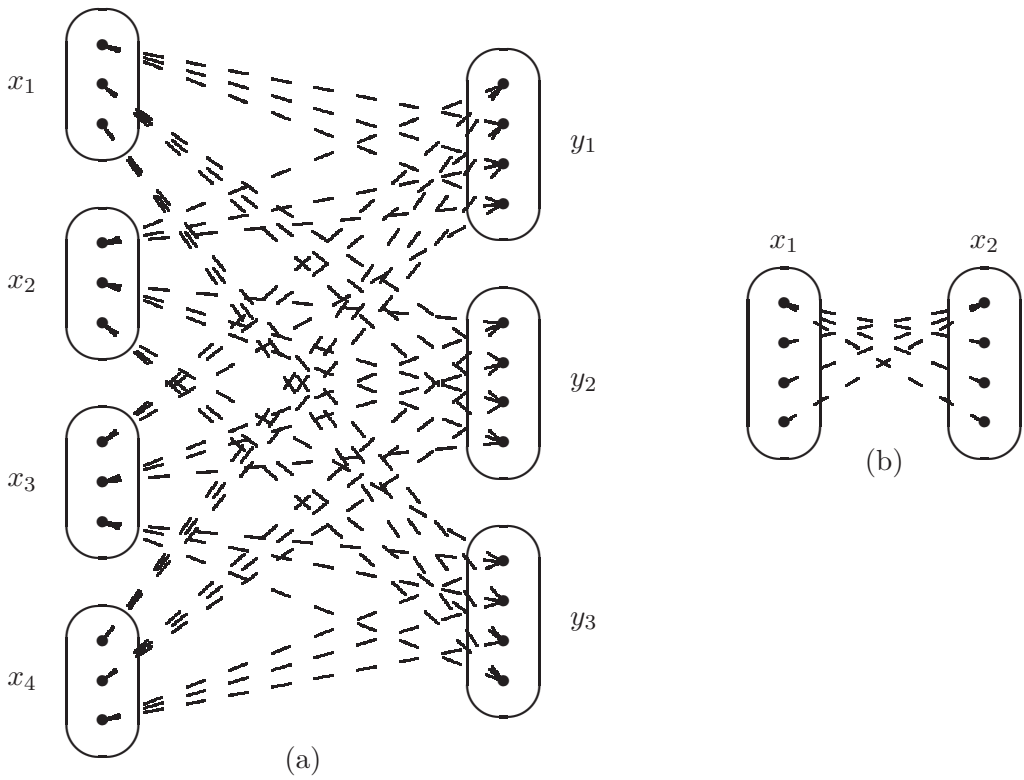

In order to establish some basic properties of patterns solvable by
SAC, we first show that several small patterns are not solvable by
SAC. In order to do this, we consider the following instances:
\begin{itemize}
\item $I_{4}^{3COL}$: \ corresponds to 3-colouring the
complete graph on 4 vertices, i.e. four variables $x_1,\ldots,x_4$
with domains $D(x_i) = \{1,2,3\}$ ($i=1,\ldots,4$) and the six
inequality constraints: $x_i \neq x_j$ ($1 \leq i < j \leq 4$).
\item $I_{3,4}$: \ corresponds to an alternative encoding of
3-colouring the complete graph on 4 vertices: three new variables
$y_1,y_2,y_3$ are introduced such that $y_j=i$ if variable $x_i$ is
assigned colour $j$. There are now seven variables ($x_1$, $x_2$,
$x_3$, $x_4$, $y_1$, $y_2$, $y_3$) with domains $D(x_i) = \{1,2,3\}$
($i=1,2,3,4$), $D(y_i) = \{1,2,3,4\}$ ($i=1,2,3$) and constraints
$(x_i = j) \Rightarrow (y_j = i)$ ($i=1,2,3,4$; $j=1,2,3$). $I_{3,4}$
is shown in Figure~\ref{fig:I34}(a) (in which only negative edges are
shown so as not to clutter up the figure).
\item $I_5$: \ five variables ($x_1,\ldots,x_5$) each with domain $\{1,2,3,4\}$ and
the constraints $(x_i = j-1)$ $\Leftrightarrow$ $(x_j=i)$ for all
$i,j$ such that $1 \leq i < j \leq 5$. One constraint of this
instance is shown in Figure~\ref{fig:I34}(b) (again only negative
edges are shown).
\end{itemize}
It is tedious but easy to verify that each of these instances has no
solution and is singleton arc consistent. Any pattern which is
solvable by SAC must therefore occur in each of these instances.
Consider the patterns shown in Figure~\ref{figNotSACpatterns}.
The patterns \PatternTOne\ and M3 do not occur in $I_{3,4}$. The pattern
Trestle does not occur in $I_5$. It therefore follows that \PatternTOne, M3 and Trestle are not solvable by SAC. Note that while M3 and Trestle are known to be NP-hard~\cite{cccms12:jair, Cooper15:dam}, the pattern T1 is tractable (but not SAC-solvable, by the argument above)~\cite{Cooper15:dam}.

\thicklines \setlength{\unitlength}{1.5pt}
\begin{figure}
\centering
\begin{picture}(240,90)(0,0)

\put(0,10){\begin{picture}(100,70)(0,0)
\put(10,60){\makebox(0,0){$\bullet$}} \dashline{5}(10,60)(50,40)
\dashline{5}(50,40)(90,60) \put(90,60){\makebox(0,0){$\bullet$}}
\put(10,60){\line(4,-3){40}} \dashline{5}(50,20)(80,10)
\put(50,20){\makebox(0,0){$\bullet$}} \put(50,20){\line(1,1){40}}
\put(50,30){\line(3,-2){30}} \put(50,30){\makebox(0,0){$\bullet$}}
\put(50,40){\makebox(0,0){$\bullet$}}
\put(80,10){\makebox(0,0){$\bullet$}} \put(10,60){\oval(18,18)}
\put(50,30){\oval(18,38)} \put(80,10){\oval(18,18)}
\put(90,60){\oval(18,18)}

\put(10,10){\makebox(0,0){M3}}
\end{picture}}

\put(140,60){\begin{picture}(70,30)(-20,0)
\put(10,10){\makebox(0,0){$\bullet$}}
\put(10,20){\makebox(0,0){$\bullet$}}
\put(40,10){\makebox(0,0){$\bullet$}} \put(10,10){\line(1,0){30}}
\put(40,20){\makebox(0,0){$\bullet$}} \put(10,10){\line(3,1){30}}
\put(10,20){\line(3,-1){30}}  \dashline{4}(10,20)(40,20)
\put(10,15){\oval(18,28)} \put(40,15){\oval(18,28)}

\put(-14,12){\makebox(0,0){Trestle}}
\end{picture}}

\put(130,0){\begin{picture}(100,40)(0,0)
\put(10,30){\makebox(0,0){$\bullet$}} \dashline{4}(10,30)(50,20)
\dashline{4}(50,20)(90,30) \put(90,30){\makebox(0,0){$\bullet$}}
\put(10,30){\line(2,-1){40}} \put(50,10){\makebox(0,0){$\bullet$}}
\put(50,10){\line(2,1){40}} \put(50,20){\makebox(0,0){$\bullet$}}
\put(10,30){\oval(18,18)} \put(90,30){\oval(18,18)}
\put(50,15){\oval(18,28)}

\put(10,10){\makebox(0,0){\PatternTOne}}
\end{picture}}

\end{picture}
\caption{Patterns not solved by SAC.} \label{figNotSACpatterns}
\end{figure}
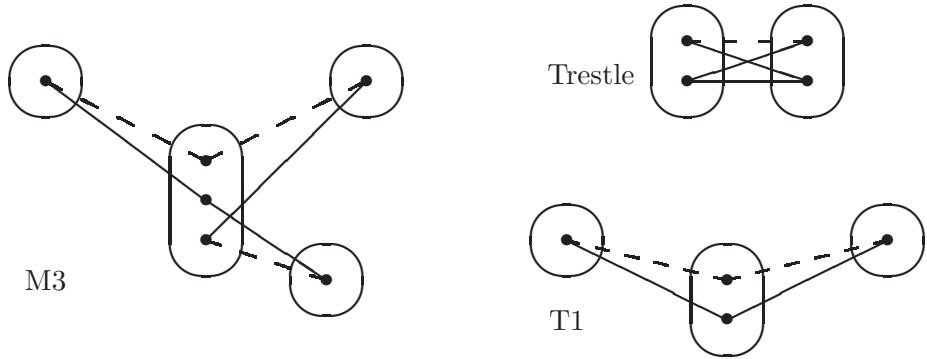

\begin{figure}
\centering

\begin{picture}(80,30)(0,0)
\put(10,15){\makebox(0,0){$\bullet$}}
\put(40,10){\makebox(0,0){$\bullet$}} \put(10,15){\line(6,1){30}}
\put(40,20){\makebox(0,0){$\bullet$}} \put(10,15){\line(6,-1){30}}
\put(10,15){\oval(18,18)} \put(40,15){\oval(18,28)}
\put(45,20){\makebox(0,0){$a$}} \put(45,10){\makebox(0,0){$b$}}
\put(70,15){\makebox(0,0){$a \neq b$}} \put(5,15){\makebox(0,0){$c$}}
\
\end{picture}

\caption{The pattern V.} \label{figV}
\end{figure}
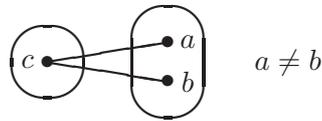

The \emph{constraint graph of a pattern} $P$ with variables $X$ is
the graph $G=(X,E)$ such that $(x,y)\in E$ if $P$ has a negative edge
between variables $x,y \in X$.

\begin{proposition} \label{prop:app}
A monotone irreducible pattern $P$ solvable by SAC satisfies:
\begin{enumerate}
\item None of the patterns \PatternTOne, M3 and Trestle occur in $P$.
\item $P$ has at most four variables. 
\item $P$ has at most one degree-3 variable and at most one non-trivial
constraint in which the pattern V, shown in Figure \ref{figV}, occurs
(with its centre point $c$ at a variable with domain at size at most
two), but does not have both a degree-3 variable and an occurrence of
V. Furthermore, $P$ has an acyclic constraint graph.
\item $P$ has at most one negative edge
per constraint, at most one point at which two negative edges meet (a
negative meet point) and no point at which three negative edges meet.
If $P$ has a negative meet point, then none of its variables has
domain size greater than two.
\end{enumerate}
\end{proposition}

\begin{proof}
The first property follows from the above discussion and the fact
that $Q$ occurs in $P$ implies that \CSP{Q} $\subseteq$ \CSP{P}.

Since a pattern $P$ which is solvable by SAC must occur in
$I_{4}^{3COL}$, $P$ can have at most four variables.
Since $P$ is also a monotone pattern, no constraint of $P$ contains only positive edges.
Since $P$ has at most four variables, either its constraint graph is connected or it is the union of
two one-constraint patterns. In the latter case, by irreducibility and because Trestle  cannot occur in $P$,
$P$ must be simply two negative edges between distinct variables (and hence all conditions of the proposition are trivially
satisfied). So we assume in the rest of the proof that the constraint graph of $P$ is connected.
Since $P$ has at most four variables, all of its variables are at a distance of at most
three in its constraint graph.

We now consider the third property.
By padding out $I_{4}^{3COL}$ with chains of equality constraints, it
is easy to produce a SAC instance which has no solution and in which
the pattern V does not occur in any non-trivial constraint at a
distance of three or less from a degree-3 variable. It follows that
no monotone irreducible pattern $P$ with a degree-3 variable and in
which the pattern V occurs is solvable by SAC. Using this same
padding-with-equality argument, we can also deduce that in a monotone
irreducible pattern $P$ solvable by SAC: there is at most one
degree-3 variable, there is at most one non-trivial constraint in
which the pattern V occurs, and that $P$ has no cycle in its
constraint graph (the latter following from the fact that cycles of
any \emph{fixed} length can be eliminated from an instance by padding out with chains
of equality constraints).

Finally, we consider the fourth property.
Each inequality constraint $x_i \neq x_j$ in $I_{4}^{3COL}$ can be
replaced by equivalent gadgets in which all constraints have at most
one negative edge~\cite{cccms12:jair}. The resulting instance is
still SAC. To be concrete, for each $i \in \{1,2,3,4\}$ and each $a
\in \{1,2,3\}$, we create 21 new Boolean variables $x_{ia}^{r}$
($r=0,1,\ldots,20$) linked to the $x_i$ variables and between
themselves by the following constraints: $(x_i=a) \Rightarrow
x_{ia}^{0}$ and $x_{ia}^{r} \Rightarrow x_{ia}^{r+1}$
($r=0,1,\ldots,19)$. 
If $x_i$ is assigned the value $a$, then all the variables $x_{ia}^{r}$ must
be assigned true. Each constraint $x_i \neq x_j$ ($1 \leq i < j \leq 4$) is then replaced by
the chain of constraints $x_{ia}^{4j} \Rightarrow y_{ija}^{1}$,
$y_{ija}^{1} \Rightarrow y_{ija}^{2}$, $y_{ija}^{2} \Rightarrow y_{ija}^{3}$,
$y_{ija}^{3} \Rightarrow \overline{x_{ja}^{4i}}$, where $y_{ija}^{s}$
($s=1,2,3$; $a=1,2,3$) are new boolean variables.
In the resulting instance $I$ there are no
points at which three negative edges meet, no two negative meet points
at a distance of three or less and no negative meet point at a distance of three
or less from a variable with domain size three.  We do not change
the semantics of $I$ (nor its singleton arc consistency) by replacing
the constraints $(x_i=a) \Rightarrow x_{ia}^{0}$ by $(x_i=a)
\Leftrightarrow x_{ia}^{0}$. In the resulting instance $I'$, no
pattern V (Figure \ref{figV}) occurs with its centre point $c$ at a
variable with domain size greater than two. We can deduce that a
monotone irreducible pattern $P$ solvable by SAC (since it contains no 
4-constraint path)
has at most one negative edge per constraint, at most one negative meet point,
no point at which three negative edges meet and the V
pattern only occurs in $P$ with its centre point $c$ at a variable
with domain of size at most two. Besides, $P$ cannot have both a
negative meet point and a variable with domain size three or more.
\end{proof}

Proposition~\ref{prop:app} allows us to narrow down monotone
irreducible patterns solvable by SAC to a finite number, which we can
summarize succinctly by the following proposition.

\begin{proposition} \label{prop:app2}
If $P$ is a monotone irreducible pattern solvable by SAC, then $P$
must occur in at least one of the patterns
\PatternPThree,\PatternQThree,R1,$\ldots$,R10 (shown in
Figures~\ref{fig:open3} and~\ref{fig:open2}).
\end{proposition}

\begin{proof}
We saw in the proof of Proposition~\ref{prop:app} that if $P$ does not have a connected constraint graph,
then $P$ is simply the union of two negative edges: in this case $P$ occurs in all the patterns R1,$\ldots$,R10.
So we assume from now on that $P$ has a connected constraint graph.
From Proposition~\ref{prop:app}, we can deduce that the constraint
graph of $P$ is a either a star or a chain, with at most four
vertices, and $P$ has at most one negative edge per constraint. 
Such patterns must have one of the following four descriptions, which we analyse separately.
\begin{description}
\item[$P$ has a single degree-3 variable. \ ]
The constraint graph of $P$ is necessarily a star. By
Proposition~\ref{prop:app}, the pattern V does not occur in $P$. 
From this and the fact that $P$ contains
no dangling points and no mergeable points, we can deduce that each of the three
degree-1 variables must have domain size 1. If the central degree-3
variable has domain size 3, then the fact that none of the patterns
V, \PatternTOne\ and M3 occur in $P$, and that there are no mergeable
points, implies that $P$ must be the pattern \PatternPThree. If, on
the other hand, the central variable has domain size 2, then since V
and \PatternTOne\ do not occur in $P$ and no three negative edges
meet at a point, we can deduce that $P$ must be \PatternQThree\ (or a subpattern).
\item[$P$ is of degree 2 and has a negative meet point. \ ]
By Proposition~\ref{prop:app}, $P$ has no domain of size greater than
2 and Trestle does not occur. It then follows by
Proposition~\ref{prop:app} and irreducibility of $P$ that the pattern
V cannot occur more than once (even in the same constraint). Since
$P$ has no dangling points, there are only four possible positions
where V could occur. We can only add a limited number of positive
edges without introducing \PatternTOne, Trestle, dangling points or
mergeable points. This gives rise to the four patterns R1,R2,R3,R4
(or subpatterns).
\item[$P$ is of degree 2, has no negative meet point and all domains have
size 1 or 2.] By the same argument as in the previous case, the
pattern V can occur at most once. By the absence of Trestle and
dangling points in $P$, and by symmetry, there are only two possible
positions for the pattern V in $P$, if it occurs at all. Again, we
can only add a limited number of positive edges without introducing
\PatternTOne, Trestle, dangling points or mergeable points. This
gives rise to the three patterns R5,R6,R10 (or subpatterns).
\item[$P$ is of degree 2, has no negative meet point and at least one size-3 domain.  ]
The fact that $P$ has no mergeable points and all variables have degree at most 2 implies that
no domain can be greater than size 3. Indeed, from the fact that $P$
is irreducible and that, by Proposition~\ref{prop:app}, no V can
occur centred at a variable of domain size 3, we can deduce that
there is exactly one variable with domain size 3. Adding positive
edges to ensure that no two points are mergeable at this variable
$v$, necessarily creates a V pattern. No other V can occur either in
a different constraint (by Proposition~\ref{prop:app}) or in the same
constraint otherwise we would have a V centred at $v$ or Trestle
would occur. Adding other positive edges, while satisfying the
properties of Proposition~\ref{prop:app}, produces patterns R7,R8,R9
(or subpatterns).
\end{description}
\end{proof}

\section{Conclusion}
\label{sec:Conclusion}

We have established SAC-solvability of five novel classes of binary CSPs defined
by a forbidden pattern, three of which are generalisations
of 2-SAT. For monotone patterns (defining classes of CSPs closed
under removing constraints), there remains only a relatively small number of  irreducible patterns whose
SAC-solvability is still open. In addition to settling the remaining patterns, a possible line of future work
is to study \emph{sets} of patterns or partially-ordered patterns~\cite{cz17:lmcs} 
that give rise to SAC-solvable (monotone) classes of CSPs.

\newcommand{\noopsort}[1]{}\newcommand{\Zivny}{\noopsort{ZZ}\v{Z}ivn\'y}

\appendix

\section{SAC-solvability of \PatternTThree}
\label{sec:T3}

Recall \PatternTThree\ from Figure~\ref{figT}.
This is the only maximal two-constraints tractable pattern whose
SAC-solvability is not determined by our main theorem,
Proposition~\ref{prop:app} and the results of~\cite{Cooper15:dam}.

\begin{theorem}
\label{thm:T3}
\CSP{\PatternTThree} is solved by singleton arc consistency.
\end{theorem}

\begin{proof}
Let $I \in$ \CSP{\PatternTThree} be a singleton arc consistent
instance with no neighbourhood substitutable values. If
\PatternTFour\ does not occur then $I$ has a solution, so we examine the
case where \PatternTFour\ occurs on variables $(x,y,z)$ and values
$a_x,a_y,b_y,c_y,a_z$ with $a_xa_y,a_xb_y,b_ya_z$ being positive
edges and $a_xc_y,a_ya_z$ being negative edges. By arc consistency,
$c_y$ has a support $b_x$ at $x$ and because \PatternTThree\ does not
occur $b_xb_y$ is a positive edge. Observe that $b_y$ dominates $c_y$
in $R(x,y)$, and because neighbourhood substitutable values have been
removed there must exist a variable $w$ (possibly equal to $z$) and
$b_w \in D(w)$ such that $b_yb_w$ is a negative edge and $c_yb_w$ is
a positive edge. However, in this case \PatternTThree\ occurs on
$(x,y,w)$, so we obtain a contradiction. It follows that
\PatternTFour\ cannot occur in the instance, and hence $I$ has a
solution.
\end{proof}

\end{document}